\documentclass[aps,prx,a4paper,superscriptaddress,longbibliography,showpacs,onecolumn,accepted=2026-04-06]{quantumarticle}
\pdfoutput=1

\usepackage{soul}
\setstcolor{red}
\usepackage{comment}

\usepackage{caption}
\DeclareCaptionJustification{justified}{\leftskip=0pt \rightskip=0pt \parfillskip=0pt plus 1fil}

\usepackage{graphicx}
\usepackage{dcolumn}
\usepackage{bm}
\usepackage{caption}
\usepackage{amsmath,amssymb,amsthm}
\usepackage{braket}
\usepackage{empheq}
\usepackage{subfig}
\usepackage{tikz}
\usepackage{url}
\usepackage{hyperref}
\usepackage{quantikz}
\usepackage{xcolor}
\usepackage{float}
\DeclareMathOperator{\Tr}{Tr}

\DeclareMathOperator{\polylog}{polylog}

\DeclareMathOperator{\algL}{\{\mathcal{L}\}_{c}}

\usepackage{algorithm}
\usepackage[noend]{algpseudocode}

\newcommand\norm[1]{\left\lVert#1\right\rVert}
\newcommand\dnorm[1]{\left\lVert#1\right\rVert_{\diamond}}

\hypersetup{
    breaklinks=true,
    colorlinks=true,       
    linkcolor=blue,          
    citecolor=red,        
    filecolor=magenta,      
    urlcolor=blue,           
    runcolor=cyan
}

\newtheorem{theorem}{Theorem}
\newtheorem*{theorem*}{Theorem}
\newtheorem{corollary}{Corollary}[theorem]
\newtheorem{lemma}[theorem]{Lemma}

\newtheorem{definition}{Definition}

\begin{document}

\title{Quantum simulation algorithms based on quantum trajectories}
\author{Evan Borras}
\affiliation{Center for Quantum Information and Control, University of New
Mexico, Albuquerque, NM 87131, USA}
\affiliation{Department of Physics and Astronomy, University of New Mexico,
Albuquerque, NM 87131, USA}
\author{Milad Marvian}
\affiliation{Center for Quantum Information and Control, University of New
Mexico, Albuquerque, NM 87131, USA}
\affiliation{Department of Physics and Astronomy, University of New Mexico,
Albuquerque, NM 87131, USA}
\affiliation{Department of Electrical \& Computer Engineering, University of New Mexico, Albuquerque, NM 87131, USA}

\begin{abstract}
Quantum simulation has emerged as a key application of quantum computing, with significant progress made in algorithms for simulating both closed and open quantum systems. 
The simulation of open quantum systems, particularly those governed by the Lindblad master equation, has received attention recently with the current state-of-the-art algorithms having an input model query complexity of $O(T\polylog(T/\epsilon))$, where $T$ and $\epsilon$ are the desired time and precision of the simulation respectively. For the Hamiltonian simulation problem it has been show that the optimal Hamiltonian query complexity is $O(T + \log(1/\epsilon))$, which is additive in the two parameters, but for Lindbladian simulation this question remains open. In this work we show that the additive query complexity to a Lindbladian's jump operators is reachable for the simulation of a large class of Lindbladians by constructing a novel quantum algorithm based on quantum trajectories.
\end{abstract}

\maketitle

\section{Introduction}

Quantum simulation has been considered an important application for quantum computation since the idea was first proposed in \cite{feynman1982simulating}. Over the years there as been an explosion of progress in designing quantum algorithms to simulate both closed quantum systems \cite{berry2007efficient,berry2013gateefficient,berry2014exponential,berry2015hamiltonian,berry2015simulating,campbell2019random,childs2019fasterquantum,childs2021theory,lloyd1996universal,low2017optimal,low2019hamiltonian,nakaji2024high,poulin2011quantum} and open quantum systems \cite{dibartolomeo2023efficient,chen2024randomized, childs2017efficient, cleve2019efficient, david2024faster,guimares2023noise,guimares2024optimized,hu2020a,joo2023commutation,kliesch2011dissipative,li2023simulating,li2023succinct,liu2024simulation,peng2024quantum,schlimgen2021quantum,schlimgen2022quantumsimulation,schlimgen2022quantumstate,suri2023twounitary,borras2025quantum,pocrnic2024quantum}. Such algorithms could have useful future applications in fields ranging from chemistry to materials science given the ubiquity of the quantum simulation problem. The importance of such applications is emphasized by the circumstance that current classical algorithms are inefficient, as well as hardness results regarding the problem of quantum simulation \cite{verstraete2009quantum}. 

Historically, quantum algorithms initially addressed the simulation of closed quantum systems, despite the fact that realistic quantum systems are, to some extent, open to an unknown environment. Quantum algorithms focused on simulating open quantum systems were explored much later with a bulk of the work focused on simulating systems described by the Lindblad master equation \cite{dibartolomeo2023efficient,chen2024randomized,childs2017efficient,cleve2019efficient,david2024faster,guimares2023noise,guimares2024optimized,hu2020a,joo2023commutation,kliesch2011dissipative,li2023simulating,liu2024simulation,peng2024quantum,schlimgen2022quantumsimulation,borras2025quantum,pocrnic2024quantum}.  The full Lindblad master equation given by
\begin{equation}
    \frac{d\rho}{dt} = -i[H, \rho] + \sum_{\mu=1}^{m}\left(L_{\mu}\rho L_{\mu}^{\dagger} - \frac{1}{2}\{L_{\mu}^{\dagger}L_{\mu}, \rho\}\right), 
\end{equation}
describes the evolution of a system coupled to a Markovian environment. The effects of the closed system dynamics are model by $-i[H, \rho]$ where $H$ is the effective Hamiltonian on the system, while the additional effects the environment has on the system are modeled by the jump operators $\{L_{\mu}\}_{\mu=1}^{m}$. Simulating dynamics generated by the Lindblad master equation, or for short Lindbladian simulation, can be seen as the natural first step into designing quantum algorithms to simulate more general open quantum systems 

For most quantum simulation problems one asks, given access to the generator of the dynamics and the time of simulation $T$, what is the minimum number of queries needed to the generator, below of which the simulation is not possible. Such results are typically called no-fast-forwarding theorems, given that the lower-bound is usually $\Omega(T)$, implying one cannot simply ``fast-forward" the quantum simulation. No-fast-forwarding results have been shown for the tasks of closed system simulation \cite{berry2013gateefficient,berry2015hamiltonian}, commonly called Hamiltonian simulation, and Lindbladian simulation in \cite{childs2017efficient,ding2024lower}. Similarly, one may also ask a related question regarding the simulation's precision $\epsilon$. For the task of Hamiltonian simulation it has been show in \cite{berry2014exponential} and \cite{berry2015hamiltonian} that the lower-bound with respect to $\epsilon$ is $\Omega(\frac{\log(1/\epsilon)}{ \log\log(1/\epsilon)})$. Naturally, this suggests an optimal lower-bound of $\Omega(T + \frac{\log(1/\epsilon)}{\log\log(1/\epsilon)})$, additive with respect to the two results, although does not rule out a multiplicative optimal query complexity. For Hamiltonian simulation this question remained open until the result of Low et al. \cite{low2017optimal} showing that the optimal Hamiltonian query complexity is additive. A similar narrative has played out for Lindbladian simulation. No-fast-forwarding results have been show for queries to the Lindbladian's jump operators and algorithms have been constructed having a jump operator query complexity scaling with respect to $\epsilon$ as $O(\frac{\log(1/\epsilon)}{\log\log(1/\epsilon)})$, assuming $T = O(1)$. The current state-of-the-art Lindbladian simulation algorithms \cite{cleve2019efficient,li2023simulating,peng2024quantum} have both a Hamiltonian and jump operator query complexity of $O(T\frac{\log(T/\epsilon)}{ \log\log(T/\epsilon)})$.

In this paper we provide evidence that the additive lower-bound is in-fact reachable for Lindbladian simulation. More specifically, for the Lindbladians satisfying $\sum_{\mu=1}^{m}L_{\mu}^{\dagger}L_{\mu}\propto \mathbb{I}$,  we construct a quantum algorithm which solves the Lindbladian simulation problem and achieves a query complexity of $O(T + \frac{\log(1/\epsilon)}{\log\log(1/\epsilon)})$ to the Lindbladian's jump operators. 
For \emph{purely} dissipative Lindbladians that satisfy this restriction, our algorithm obtains a query complexity that is optimal with respect to the time of simulation.

The approach we take in constructing the algorithm is heavily inspired by the quantum trajectories/Monte Carlo wavefunction approach towards classically simulating the Lindblad master equation \cite{breuer2007the}. Simulating the quantum trajectories as opposed to the full Lindblad master equation allows us to exploit the fact that quantum jumps are timed by a Poisson clock, occurring infrequently with an average number of jumps per trajectory being $O(T)$. This allows us to reduce the number of queries to the Lindbladian's jump operators and is the core reason behind the additive query complexity.

In addition, we provide a characterization of the set of Lindbladians which satisfy our algorithm's constraint. The results we provide suggest that extending our algorithm to Lindbladians outside this set will require more than a naive black-box approach.

The paper is divided into four sections, with section~\ref{sec:preq} being a preliminaries section where we introduce some notation, algorithmic tools, and review the quantum jump/Monte Carlo wavefunction unraveling of the Lindblad master equation. In section~\ref{sec:newalg} we introduce the new algorithm and provide its resource analysis. In the section~\ref{sec:char} which follows, we provide a short argument regarding our algorithm's optimality with respect to the time of simulation $T$ for the restricted class of Lindbladian dissipators we consider, along with a brief analysis and characterization of this restricted class. Finally, we conclude with section~\ref{sec:disc} summarizing and discussing our results. 

\section{Prerequisites}
\label{sec:preq}
We begin by defining some notation. We will be working with systems of $n$-qubits implying that general quantum states $\rho$ are given by density operators, trace normalized positive linear operators acting on $\mathbb{C}^{2^{n}}$. Pure quantum states will either be denoted as rank one density operators or normalized vectors in $\mathbb{C}^{2^{n}}$. We will use the term superoperator to refer to a linear map acting on the set of bounded linear operators over $\mathbb{C}^{2^{n}}$. Thus quantum channels are superoperators which are completely positive and trace preserving. The sets of completely positive trace preserving superoperators and unital completely positive superoperators will be denoted $\mathbf{CPTP}$ and $\mathbf{UCP}$ respectively. The identity superoperator will be denoted as $\mathcal{I}$.

The operator norm of a linear operator will be denoted as $\norm{A}$. Similarly, the trace norm of a linear operator $A$ is defined as $\norm{A}_{1} = \mathrm{Tr}(\sqrt{A^{\dagger}A})$. Given a superoperator $\mathcal{L}$, the diamond norm of $\mathcal{L}$ is given by $\dnorm{\mathcal{L}} = \sup_{\rho}\norm{\mathcal{I}\otimes \mathcal{L}(\rho)}_{1}$ where the supremum is taken over all density operators $\rho$. The diamond norm of superoperators induces a diamond distance between superoperators given by $d_{\diamond}(\mathcal{E}, \mathcal{L}) = \frac{1}{2}\dnorm{\mathcal{E} - \mathcal{L}}$. We will also be using the total variational distance between probability distributions defined for continuous variables as $\delta(p,f) = \frac{1}{2}\int dx |p(x) - f(x)|$ and discrete variables as $\delta(p, f) = \frac{1}{2}\sum_{N} |p_{N} - f_{N}|$.

Finally we denote the total simulation time requested to be $T$ and the precision of the simulation to be $\epsilon$. We analyze the precision of the simulation with respect to the diamond distance given that it not only provides a measure of distance between quantum channels, it also is used to upperbound the single shot probability of distinguishing between quantum channels \cite{watrous_2018}. Thus the diamond distance provides a relevant operational measure of similarity between quantum channels. In the sections to follow we introduce some tools and techniques we will use in building our algorithm as well as review the quantum jump unraveling of the Lindblad master equation \cite{breuer2007the}.

\subsection{Algorithmic Tools}

As input, our algorithm needs some way to access the Lindbladian of the system we want to simulate. We choose to provide our algorithm access to the system's Lindbladian through the use of block-encodings of the system Lindbladian's jump operators $\{L_{\mu}\}_{\mu=1}^{m}$ and Hamiltonian $H$. Intuitively, the technique of block encoding gives a quantum algorithm access to an arbitrary linear operator $A$, that may or may not be unitary, through encoding it in the upper-left block of a unitary as follows
\begin{equation}
    U(A) = \begin{bmatrix}
        A & * \\
        * & * 
    \end{bmatrix}.
\end{equation}
The $*$ entries are arbitrary block matrices needed to make the overall $U(A)$ unitary. More formally we define the notion of block-encoding below.
\begin{definition}
    Let $A$ be an $n$-qubit operator, $\alpha \in \mathbb{R}_{+} \setminus \{0\}$ and $a \in \mathbb{N}$. We call an $(n + a)$-qubit unitary $U(A)$ an $(\alpha, a, \epsilon)$-block-encoding of $A$ if 
    \begin{equation}
        \norm{A - \alpha(\bra{0}^{\otimes a} \otimes \mathbb{I}^{\otimes n})U(A)(\ket{0}^{\otimes a}\otimes \mathbb{I}^{\otimes n})} \leq \epsilon
    \end{equation}
\end{definition}

It is well known that an aribitrary superoperator $\mathcal{M}$ which is a completely positive map can be represented in operator sum notation as $\mathcal{M}(\rho) = \sum_{\mu} A_{\mu} \rho A_{\mu}^{\dagger}$ where the linear operators $A_{\mu}$ are typically called Kraus operators \cite{Nielsen_Chuang_2010}. Given access to block-encodings $U(A_{\mu})$ for each of the operators $A_{\mu}$, the circuit construction and Lemma below, describe a method for implementing the map $\mathcal{M}$ on an input quantum state $\rho$ with some probability of failure.
\begin{figure}[h]
    \centering
    \begin{quantikz}
        \lstick{$\ket{0}$} & \qwbundle{\log(m)} & & \gate{B} & \ctrl{1} & & \\
        \lstick{$\ket{0}$} & \qwbundle{a} & & & \gate[2]{U(A_{\mu})} &  &\\
        \lstick{$\rho$} & \qwbundle{n} & & & & & \\ 
    \end{quantikz}
    \caption{Implementing CP-maps with block encodings}
    \label{fig:cpbc}
\end{figure}
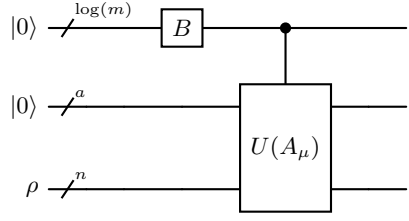
\begin{lemma}[Implementing CP-maps with  block-encodings \cite{li2023succinct}]
    \label{lm:cpbc}
    Let $\{A_{\mu}\}_{\mu=1}^{m}$ be the Kraus operators of a completely positive map $\mathcal{M}$,
    \begin{equation}
        \mathcal{M}(\rho) = \sum_{\mu=1}^{m}A_{\mu} \rho A_{\mu}^{\dagger}.
    \end{equation}
    Let $U(A_{\mu})$ be the corresponding $(\alpha_{\mu}, a, \epsilon)$-block-encoding for Kraus operator $A_{\mu}$. Define oracle $B$ and state $\ket{\omega}$  such that
    \begin{equation}
        \ket{\omega} = \frac{1}{\sqrt{\sum_{\mu=1}^{m}\alpha_{\mu}^{2}}}\sum_{\mu=1}^{m}\alpha_{\mu}\ket{\mu} = B\ket{0}
    \end{equation}
    Then the circuit depicted in FIG~\ref{fig:cpbc} implements the map $\mathcal{M}$ in the sense that for all $\ket{\psi}$ 
    \begin{align}
        \bigg\lVert \mathbb{I} \otimes \bra{0} \otimes \mathbb{I} \left(\sum_{\mu=1}^{m}\ket{\mu}\bra{\mu}\otimes U(A_{\mu})\right)&(B \otimes \mathbb{I} \otimes \mathbb{I})\ket{0} \otimes \ket{0} \otimes \ket{\psi} - \frac{1}{\sqrt{\sum_{\mu=1}^{m}\alpha_{\mu}^{2}}}\sum_{\mu=1}^{m}\ket{\mu}\otimes A_{\mu}\ket{\psi} \bigg\rVert \nonumber \\
        &\leq \frac{m\epsilon}{\sqrt{\sum_{\mu=1}^{m}\alpha_{\mu}^{2}}}.
    \end{align}
\end{lemma}

The final technique our algorithm will leverage is oblivious amplitude amplification. Given access to some target unitary/isometry through the use of a unitary $W$ implementing the target operation with some probability of failure, oblivious amplitude amplification gives us a method to boost the success probability of implementing the target operation. Thus, using oblivious amplitude amplification we can turn a faulty implementation of some target operation into a non-faulty implementation. More formally the technique is described by the Lemma below.
\begin{lemma}[Oblivious amplitude amplification for isometries \cite{li2023succinct}]
    \label{lm:oaa}
    For any $a, b \in \mathbb{N}_{+}$, let $\ket{\hat{0}} = \ket{0}^{\otimes a}$ and $\ket{\hat{\mu}} = \ket{\mu}^{\otimes b}$ for an arbitrary state $\ket{\mu}$.  For any $n$-qubit state $\ket{\psi}$, define $\ket{\hat{\psi}} = \ket{\hat{0}}\otimes \ket{\hat{\mu}} \otimes \ket{\psi}$. Let the target state $\ket{\hat{\phi}}$ be defined as 
    \begin{equation}
        \ket{\hat{\phi}} = \ket{\hat{0}}\otimes \ket{\phi}
    \end{equation}
    where $\ket{\phi}$ is a $(b + n)$-qubit state. Let $P_{0} = \ket{\hat{0}}\bra{\hat{0}}\otimes \mathbb{I}\otimes\mathbb{I}$ and $P_{1} = \ket{\hat{0}}\bra{\hat{0}}\otimes \ket{\hat{\mu}}\bra{\hat{\mu}} \otimes \mathbb{I}$ be two projectors. Suppose there exists an operator $W$ such that 
    \begin{equation}
        W\ket{\hat{\psi}} = \sin(\theta)\ket{\hat{\phi}} + \cos(\theta)\ket{\hat{\phi^{\perp}}},
    \end{equation}
    for some $\theta \in [0, \pi / 2 ]$ with $\ket{\hat{\phi^{\perp}}}$ satisfying $P_{0}\ket{\hat{\phi^{\perp}}}= 0$. Then it holds that 
    \begin{equation}
        -W(\mathbb{I} - 2P_{1})W^{\dagger}(\mathbb{I} - 2P_{0})(\sin(\gamma)\ket{\hat{\phi}} + \cos(\gamma)\ket{\hat{\phi^{\perp}}}) = \sin(\gamma + 2\theta)\ket{\hat{\phi}} + \cos(\gamma + 2\theta)\ket{\hat{\phi^{\perp}}},
    \end{equation}
    for all $\gamma \in [0, \pi / 2]$. 
\end{lemma}

While unitary $W$ implements the state $\ket{\phi}$ only faultily, with the probability of success being $\sin^{2}(\theta)$, using Lemma~\ref{lm:oaa} we can iterate the construction $-W(\mathbb{I} - 2P_{1}) W^{\dagger}(\mathbb{I} - 2P_{0})$ boosting the probability of implementing $\ket{\phi}$.

\subsection{Quantum Jump Unraveling Of The Lindblad Master Equation}

For convenience we re-state the full Lindblad master equation given below,
\begin{equation}
    \frac{d\rho}{dt}  = -i [H, \rho] + \sum_{i=1}^{m}\left(L_{\mu}\rho L_{\mu}^{\dagger} - \frac{1}{2}\{L_{\mu}^{\dagger}L_{\mu}, \rho\}\right) =  \mathcal{L}(\rho)
\end{equation}
where $\mathcal{L}$ is defined as the Lindbladian superoperator for the system.
The standard way to interpret the Lindblad master equation is that the Lindbladian superoperator is a generator for a family of completely positive trace preserving maps acting on density matrices parameterized by time. Given that density operators can also be interpreted as a probability distribution over pure states in Hilbert space leads to another interpretation of the dynamics generated by the Lindblad master equation. Alternatively, we can interpret the density operator as some initial probability distribution over Hilbert space evolving under some stochastic process \cite{breuer2007the}. The Lindblad master equation is then ``unraveled" as a stochastic Schr{\"o}dinger equation describing how the random pure state vector evolves under the stochastic process defined by the Lindbladian. Averaging over many solutions of stochastic Schr{\"o}dinger equation, called trajectories, reproduces the dynamics of a density operator evolving under the Lindblad master equation.

For a given Lindbladian $\mathcal{L}$ there exists multiple different unravelings, with the quantum jump unraveling being the one we will base our algorithm off of \cite{breuer2007the}.
In this picture, the dynamics generated by the Lindblad master equation can be interpreted as a piece-wise deterministic stochastic process over Hilbert space and the quantum state in operator form can be reconstructed as the expectation value of the pure state operator $\ket{\psi}\bra{\psi}$ with respect to this distribution at the given time. In this unraveling, a pure state of the ensemble associated with the initial density operator evolves deterministically under a non-Hermitian operator given by $\hat{H} = H - \frac{i}{2}\sum_{\mu = 1}^{m}L_{\mu}^{\dagger}L_{\mu}$ for a random period of time before being interrupted by quantum jumps. The deterministic evolution is given by the map
\begin{equation}
    \label{eq:detevo}
    G_{t}(\ket{\psi}) = \frac{e^{-i\hat{H}t}\ket{\psi}}{\norm{e^{-i\hat{H}t}\ket{\psi}}},
\end{equation}
along with the $\mu$th quantum jump given by
\begin{equation}
    \label{eq:jump}
    J_{\mu}(\ket{\psi}) = \frac{L_{\mu}\ket{\psi}}{\norm{L_{\mu}\ket{\psi}}} \quad \text{ occuring with probability}
    \quad \mathrm{Pr}(\mu | \ket{\psi}) = \frac{\norm{L_{\mu}\ket{\psi}}^{2}}{\sum_{\mu = 1}^{m}\norm{L_{\mu}\ket{\psi}}^{2}}.
\end{equation}

Given an initial state $\ket{\psi(0)}$ evolving under a trajectory interrupted by a total of $N$ quantum jumps, the resulting state at time $T$ can be expressed as:
\begin{equation}
    \label{eq:extraj}
    \ket{\psi(T)} = G_{(T - t^{(N)})} \circ J_{\mu^{(N)}} \circ G_{(t^{(N)} - t^{(N-1)})} \cdots J_{\mu^{(2)}}\circ G_{(t^{(2)} - t^{(1)})}\circ J_{\mu^{(1)}}\circ G_{t^{(1)}}(\ket{\psi(0)})
\end{equation}
where each random $t^{(i)} \in [0, T]$ indicates the time of the $i$-th quantum jump. Defining the $i$-th holding-time between quantum jumps to be $t_{h}^{(i)} = t^{(i)} - t^{(i-1)}$ allows us to rewrite equation~\eqref{eq:extraj} as
\begin{align}
    \label{eq:trajop}
    \ket{\psi(T)} &= G_{(T - \sum_{i=1}^{N}t_{h}^{(i)})} \circ J_{\mu^{(N)}} \circ G_{t_{h}^{(N)}} \cdots J_{\mu^{(2)}}\circ G_{t_{h}^{(2)}}\circ J_{\mu^{(1)}}\circ G_{t_{h}^{(1)}}(\ket{\psi(0)}) \nonumber \\
                  &= P_{t_{h}^{(1)}, \mu^{(1)}, \cdots, t_{h}^{(N)}, \mu^{(N)}}(\ket{\psi(0)})
\end{align}
where we have defined the trajectory operator $P_{t_{h}^{(1)}, \mu^{(1)}, \cdots, t_{h}^{(n)}, \mu^{(n)}}$ to be the map which evolves the input state $\ket{\psi(0)}$ along the sampled trajectory.
Since quantum jumps occur at random intervals of time, the holding times $t_{h}^{(i)}$ are also defined as random variables, with the conditional cumulative distribution function given by
\begin{equation}
    \label{eq:holdingtime}
    \mathrm{Pr}\left(t_{h} \leq \tau | \ket{\psi(t)}\right)= 1 - \norm{e^{-i \tau \tilde{H}}\ket{\psi(t)}}^{2}.
\end{equation}

Evolving an initial state $\ket{\psi(0)}$ under such sampled trajectory operators realizes a solution to the stochastic Schr{\"o}dinger equation associated with the quantum jump unraveling. Thus, averaging over all trajectories associated with the initial state $\ket{\psi(0)}$ reproduces the exact evolution of the initial density operator $\ket{\psi(0)}\bra{\psi(0)}$ under the Lindblad master equation. Defining the superoperator form of the trajectory operator as
\begin{equation}
    \mathcal{P}_{t_{h}^{(1)}, \mu^{(1)}, \cdots, t_{h}^{(N)}, \mu^{(N)}}(\ket{\psi}\bra{\psi}) = P_{t_{h}^{(1)}, \mu^{(1)}, \cdots, t_{h}^{(N)}, \mu^{(N)}}(\ket{\psi})P_{t_{h}^{(1)}, \mu^{(1)}, \cdots, t_{h}^{(N)}, \mu^{(N)}}(\ket{\psi})^{\dagger}
\end{equation}
allows us to write
\begin{align}
    \label{eq:purestateqtraj}
    e^{T\mathcal{L}}\left(\ket{\psi(0)}\bra{\psi(0)}\right) &= \mathbb{E}\left[\ket{\psi(T)}\bra{\psi(T)}\right] \nonumber \\
                                                            &= \sum_{N=1}^{\infty}\int_{0}^{T}dt_{h}^{(1)}\sum_{\mu^{(1)}=1}^{m}\cdots \int_{0}^{T}dt_{h}^{(N)}\sum_{\mu^{(N)}=1}^{m} \nonumber \\ &p(t_{h}^{(1)}, \mu^{(1)},\cdots, t_{h}^{(N)}, \mu^{(N)}; T) \mathcal{P}_{t_{h}^{(1)}, \mu^{(1)},\cdots, t_{h}^{(N)}, \mu^{(N)}}(\ket{\psi(0)}\bra{\psi(0)}).
\end{align}
where $p(t_{h}^{(1)}, \mu^{(1)}, \cdots, t_{h}^{(N)}, \mu^{(N)}; T) = \mathrm{Pr}\left(t_{h}^{(1)}, \mu^{(1)}, \cdots, t_{h}^{(N)}, \mu^{(N)}, \sum_{i=1}^{N}t_{h}^{(i)} \leq T, \sum_{i=1}^{N+1}t_{h}^{(i)} > T\right)$  denotes the probability density for sampling a specific trajectory with final evolution time $T$, holding times $\{t_{h}^{(i)}\}_{i=1}^{N}$ and quantum jumps $\{\mu^{(i)}\}_{i=1}^{N}$. Averaging the initial pure state $\ket{\psi(0)}$ over the distribution associated with the initial density operator $\rho(0)$ gives us
\begin{align}
    \label{eq:qtraj}
    e^{T\mathcal{L}}\left(\rho(0)\right) &= \sum_{N=1}^{\infty}\int_{0}^{T}dt_{h}^{(1)}\sum_{\mu^{(1)}=1}^{m}\cdots \int_{0}^{T}dt_{h}^{(N)}\sum_{\mu^{(N)}=1}^{m} \nonumber \\ &p(t_{h}^{(1)}, \mu^{(1)},\cdots, t_{h}^{(N)}, \mu^{(N)}; T) \mathbb{E}_{\psi \sim \rho(0)}\left[\mathcal{P}_{t_{h}^{(1)}, \mu^{(1)},\cdots, t_{h}^{(N)}, \mu^{(N)}}(\ket{\psi(0)}\bra{\psi(0)})\right], 
\end{align}
which relates the dynamics generated by Lindbladian $\mathcal{L}$ to an average over trajectories on the pure state space.

In the section to follow we introduce a new quantum algorithm inspired by the framework described above restricted to Lindbladians which satisfy the constraint,
\begin{equation}
    \label{eq:algconstraint}
    \sum_{\mu=1}^{m}L_{\mu}^{\dagger}L_{\mu} = \Gamma \mathbb{I},
\end{equation}
for some $\Gamma \in \mathbb{R}_{+}$. Although our algorithm is not fully general, many common Lindbladians satisfy the above constraint such as Pauli channel dissipation, and  Lindbladians used to prepare certain thermal states\cite{chen2023quantum,chen2023efficient}. Moreover Lindbladians used to prepare ground states of frustration free Hamiltonians \cite{verstraete2009quantum}, as well as Lindbladians that perform continuous-time error correction \cite{Oreshkov_2013} also satisfy equation~\eqref{eq:algconstraint}. In-addition, any classical continuous-time Markov chain can be embedded into a Lindbladian which satisfies equation~\eqref{eq:algconstraint}; to do so one embeds the transition matrix associated to the continuous-time Markov chain into its quantum channel equivalent and the rate transitions occur at into $\Gamma$. Lemma~\ref{lm:ntos} of Appendix~\ref{app:A} implies that our Lindbladian which satisfies equation~\eqref{eq:algconstraint} is in-turn fully defined by this quantum channel and rate $\Gamma$. Interestingly certain models of cross-talk in superconducting qubits have been modeled by Lindbladians that satisfy our algorithm's constraint as well \cite{tripathi2022suppression}.

\section{New Algorithm}
\label{sec:newalg}

At a high level our algorithm involves sampling the random times of deterministic evolution $t_{h}^{(i)}$ during compilation, and building a quantum circuit that implements the trajectory operator on the initial state. Equation~\eqref{eq:qtraj} assures the measurement statistics produced at the end of each run of a sampled circuit reproduces the measurement statistics of evolving the initial state under the map $e^{T\mathcal{L}}$. As inputs to our classical compilation procedure we will assume access to the quantum circuit constructions for an $(\alpha, a, 0)$-block-encoding $U(H)$ of the Hamiltonian associated with $\mathcal{L}$, an $(\alpha_{\mu}, a, 0)$-block-encoding $U(L_{\mu})$ for each of the $m$ jump operators $L_{\mu}$ associated with $\mathcal{L}$, as well as their respective inverses. Before discussing the full circuit compilation procedure, we first explore how our algorithm's constraint of equation~\eqref{eq:algconstraint} simplifies the quantum jump unraveling of the Lindblad master equation discussed above.

To see how the jump unraveling simplifies, consider a pure state out of an ensemble which defines the density operator for a state at time $t$. Denote this pure state by $\ket{\psi(t)}$. Equation~\eqref{eq:holdingtime}, which defines the cumulative distribution for the holding time between the current state and next jump $t_{h}$, gives 
\begin{align}
    \label{eq:algcdf}
    \mathrm{Pr}(t_{h} \leq \tau | \ket{\psi(t)}) &= 1 - \norm{e^{-i \tau \tilde{H}}\ket{\psi(t)}}^{2} = 1 - \norm{e^{-i\tau H - \frac{\tau}{2}\sum_{\mu=1}^{m}L_{\mu}^{\dagger}L_{\mu}}\ket{\psi(t)}}^{2} \nonumber \\
                                              &= 1 - \norm{e^{-i \tau H - \frac{\tau \Gamma}{2} \mathbb{I}}\ket{\psi(t)}}^{2}  = 1 - e^{-\tau \Gamma}\norm{e^{-i \tau H} \ket{\psi(t)}}^{2} \nonumber \\
            & = 1 - e^{-\Gamma \tau},
\end{align}
where the second line follows from applying the constraint of equation~\eqref{eq:algconstraint}.  The key point to note about equation~\eqref{eq:algcdf} is that the resulting cumulative distribution function is \emph{independent} of $\ket{\psi(t)}$. This implies we can sample from both $\mathrm{Pr}(t_{h} \leq \tau | \ket{\psi(t)})$ and $\mathrm{Pr}(t_{h} = \tau | \ket{\psi(t)})$ without knowing the state $\rho(t)$ or which pure state in the ensemble we are evolving. This condition is what allows us to sample the random times of deterministic evolution during the compilation of a circuit implementing a trajectory.

Next, we consider the types of trajectory operators associated with  Lindbladians which statisfy equation~\eqref{eq:algconstraint}. Consider the portion of a trajectory characterized by deterministic evolution given by equation~\eqref{eq:detevo}. Applying our constraint of equation~\eqref{eq:algconstraint} allows us to write the map characterizing the deterministic evolution as:
\begin{align}
    \label{eq:constraintdetevo}
    G_{t}(\ket{\psi}) &= \frac{e^{-it\hat{H}}\ket{\psi}}{\norm{e^{-it\hat{H}}\ket{\psi}}} = \frac{e^{-itH - \frac{t}{2}\sum_{\mu=1}^{m}L_{\mu}^{\dagger}L_{\mu}}\ket{\psi}}{\norm{e^{-itH - \frac{t}{2}\sum_{\mu=1}^{m}L_{\mu}^{\dagger}L_{\mu}}\ket{\psi}}} = \frac{e^{-t\frac{\Gamma}{2}} e^{-itH} \ket{\psi}}{\norm{e^{-t\frac{\Gamma}{2}} e^{-itH} \ket{\psi}}} = e^{-itH}\ket{\psi}.
\end{align}
Thus the deterministic evolution of equation~\eqref{eq:detevo} reduces to generic unitary evolution under the Hamiltonian $H$ of the Lindbladian. Considering the jump dynamics next, and applying our constraint of equation~\eqref{eq:algconstraint}, allows us to write down the jump dynamics for Lindbladians that satisfy equation~\eqref{eq:algconstraint} as:
\begin{equation}
    \label{eq:constraintjumpevo}
    J_{\mu}(\ket{\psi}) = \frac{L_{\mu}\ket{\psi}}{\norm{L_{\mu}\ket{\psi}}} \quad \text{ occuring with probability}
    \quad \mathrm{Pr}(\mu | \ket{\psi}) = \frac{\norm{L_{\mu}\ket{\psi}}^{2}}{\Gamma}.
\end{equation}
Defining $\mathcal{J}_{\mu}$ to be the superoperator form of the jump maps defined above allows us to average over all possible jump maps, finding that that the average map
\begin{align}
    \label{eq:jumpevo}
    \mathcal{J}(\ket{\psi}\bra{\psi}) = \sum_{\mu=1}^{m}\mathrm{Pr}(\mu|\ket{\psi})\mathcal{J}_{\mu}(\ket{\psi}\bra{\psi}) = \sum_{\mu=1}^{m}\mathrm{Pr}(\mu|\ket{\psi})J_{\mu}(\ket{\psi})J_{\mu}(\ket{\psi})^{\dagger} = \sum_{\mu=1}^{m}\frac{L_{\mu}\ket{\psi}\bra{\psi}L_{\mu}^{\dagger}}{\Gamma},
\end{align}
which is  a valid quantum channel with Kraus operators $L_{\mu} / \sqrt{\Gamma}$. If we define the map $\mathcal{U}_{t}(\rho) = e^{-itH}\rho e^{itH}$, then we can express an arbitrary trajectory operator for a Lindbladian that satisfies equation~\eqref{eq:algconstraint} as
\begin{equation}
    \label{eq:constrainttraj}
    \mathcal{P}_{t_{h}^{(1)}, \mu^{(1)}, \cdots, t_{h}^{(N)}, \mu^{(N)}}(\ket{\psi}\bra{\psi}) = \mathcal{U}_{(T - \sum_{i=1}^{N}t_{h}^{(i)})} \circ \mathcal{J}_{\mu^{(N)}} \circ \mathcal{U}_{t_{h}^{(N)}} \cdots \mathcal{J}_{\mu^{(1)}} \circ \mathcal{U}_{t_{h}^{(1)}}(\ket{\psi}\bra{\psi}).
\end{equation}
Thus, sampled trajectories for Lindbladians which satisfy equation~\eqref{eq:algconstraint} are built from generic unitary evolution interrupted by quantum jumps. Pushing this analysis further, we can substitute equation~\eqref{eq:constrainttraj} into equation~\eqref{eq:qtraj} and perform the sums over $\{\mu^{(i)}\}_{i=1}^{N}$. Applying equation~\eqref{eq:jumpevo} gives us the final form for the description of the evolution under $\mathcal{L}$ satisfying equation~\eqref{eq:algconstraint} as
\begin{align}
    \label{eq:constraintqtraj}
    e^{T\mathcal{L}}\left(\rho(0)\right) &= \sum_{N=1}^{\infty}\int_{0}^{T}dt_{h}^{(1)}\cdots \int_{0}^{T}dt_{h}^{(N)} p(t_{h}^{(1)}, \cdots, t_{h}^{(N)}; T) \mathcal{U}_{(T-\sum_{i=1}^{N}t_{h}^{(i)})}\circ \mathcal{J} \circ \mathcal{U}_{t_{h}^{(N)}} \cdots \mathcal{J} \circ \mathcal{U}_{t_{h}^{(1)}}(\rho(0))\nonumber \\
                                         &= \sum_{N=1}^{\infty}\int_{0}^{T}dt_{h}^{(1)}\cdots \int_{0}^{T}dt_{h}^{(N)} p(t_{h}^{(1)}, \cdots, t_{h}^{(N)}; T) \mathcal{T}_{t_{h}^{1}, \cdots, t_{h}^{(N)}}(\rho(0)),
\end{align}
where $p(t_{h}^{(1)}, \cdots, t_{h}^{(N)}; T)$ denotes the probability of sampling a trajectory that occurs over time $T$ with holding times $\{t_{h}^{(i)}\}_{i=1}^{N}$  and $\mathcal{T}_{t_{h}^{(1)}, \cdots, t_{h}^{(N)}} = \mathcal{U}_{(T-\sum_{i=1}^{N}t_{h}^{(i)})}\circ \mathcal{J} \circ \mathcal{U}_{t_{h}^{(N)}} \cdots \mathcal{J} \circ \mathcal{U}_{t_{h}^{(1)}}$. A trajectory of $N$ jumps occurring over time $T$ implies the sum over the $N$ holding times $t_{h}^{(i)}$ is less than or equal to the final time of the trajectory $T$ and adding an additional $N+1$ holding time makes this sum greater than the final time of the trajectory. Defining $S^{(N)} = \sum_{i=1}^{N}t_{h}^{(i)}$ we can write 
\begin{align}
    \label{eq:pdensity}
    p(t_{h}^{(1)}, \cdots, t_{h}^{(N)}; T) &= \mathrm{Pr}\left(t_{h}^{(1)}, \cdots, t_{h}^{(N)}, S^{(N)} \leq T, S^{(N+1)} > T\right) \nonumber \\
                                           &= \mathrm{Pr}\left(t_{h}^{(1)}, \cdots, t_{h}^{(N)} \big| S^{(N)} \leq T, S^{(N+1)} > T\right)\mathrm{Pr}\left(S^{(N)} \leq T, S^{(N+1)} > T\right) \nonumber \\
                                           &= \mathrm{Pr}\left(t_{h}^{(1)}, \cdots, t_{h}^{(N)} \big| S^{(N)} \leq T, S^{(N+1)} > T\right)\left[\mathrm{Pr}\left(S^{(N)} \leq T \right) + \mathrm{Pr}\left(S^{(N+1)} > T\right) - 1 \right]\nonumber \\
                                           &= \mathrm{Pr}\left(t_{h}^{(1)}, \cdots, t_{h}^{(N)} \big| S^{(N)} \leq T, S^{(N+1)} > T\right)\left[\mathrm{Pr}\left(S^{(N)} \leq T \right) - \mathrm{Pr}\left(S^{(N+1)} \leq T\right)\right]\nonumber \\
                                           &= \mathrm{Pr}\left(t_{h}^{(1)}, \cdots, t_{h}^{(N)} \big| S^{(N)} \leq T, S^{(N+1)} > T\right)\mathrm{Pr}(N(T) = N)
\end{align}
where the last equality follows from Lemma~\ref{lm:ntos} of Appendix~\ref{app:B} with $N(T) = \sup\{n : S^{(n)} \leq T\}$. The random variable $N(T)$ counts the total number of quantum jumps occurring with-in an interval of time $T$. Since the holding times between jump events $t_{h}^{(i)}$ have a cumulative distribution of equation~\eqref{eq:algcdf}, Theorem~\ref{thm:pois} of Appendix~\ref{app:B} implies $N(T) \sim \mathrm{Pois}(\Gamma T)$. Thus equation~\eqref{eq:pdensity} gives us a partial characterization of the probability density $p(T)$ along with some global structure regarding how trajectories evolving up to time $T$ behave.

Equations~\eqref{eq:constraintqtraj} and~\eqref{eq:pdensity} suggest we can simulate the evolution of a system under Lindbladian $\mathcal{L}$ which satisfies equation~\eqref{eq:algconstraint} by sampling a set of holding times $\mathcal{S} = \{t_{h}^{(i)}\}_{i=1}^{N}$, checking if our sampled trajectory reaches the requested simulation time $T$, and then compiling a circuit composed of unitary evolution under $H$ interspersed by quantum channels $\mathcal{J}$. 

\begin{algorithm}[hbt!]
    \caption{Trajectory Compilation Protocol}\label{alg:one}
    \begin{algorithmic}
        \Require time $T$, precision $\epsilon$, query access to block-encodings $U(H)$ and $\{U(L_{\mu})\}_{\mu=1}^{m}$ and inverses, rate $\Gamma$
        \State $r \gets r(T, \epsilon, \Gamma)$
        \State $S \gets \{\} $
        \State $t \gets 0$
        \State $i \gets 1$
        \While{$|S| \leq r$ and $t \leq T$} 
        \State $t_{h}^{(i)} \sim \mathrm{Pr}(t_{h})$
        \State Append $t_{h}^{(i)}$ to $S$
        \State $t \gets t + t_{h}^{(i)}$
        \State $i \gets i + 1$
        \EndWhile
        \If{$|S| > r$ and $t < T$}
        \State Repeat \textbf{Algorithm 1}
        \Else
        \State Compile the quantum circuit in FIG~\ref{fig:circuit} associated with jump times $S \setminus \{t_{h}^{(i)}\}$ and precision $\epsilon$.
        \State
        \EndIf
    \end{algorithmic}
\end{algorithm}

The full trajectory compilation protocol described in Algorithm~\ref{alg:one} ultimately amounts to sampling deterministic evolution times $t_{h}^{(i)}$ from probability density $\mathrm{Pr}(t_{h})$ until either the sum over the samples is greater than the total simulation time $T$ or the maximum number of samples $r$ has been reached. The quantum circuit associated with the sampled trajectory is then compiled or the procedure will have to be repeated. Note, if the sum of the sampled deterministic evolution times $t_{h}^{(i)}$ is strictly greater than the total simulation time $T$, we discard the last sample to account for the fact that the trajectory has scheduled the last jump to occur after the total simulation time has elapsed.

The parameter $r$ in our algorithm denotes the maximum number of samples able to be drawn and is directly related to the maximum number of quantum jumps allowed in a compiled trajectory. This parameter is ultimately set by the total simulation time $T$ and precision of the simulation $\epsilon$.

To generate a sample $t_{h}^{(i)}$ we employ the standard protocol of inversion sampling. To inversion sample from an arbitrary probability density $p(x)$, one first samples $\eta$ from $\mathrm{uniform}(0,1)$, then calculates $x = f^{-1}(\eta)$, where $f(x)$ is the cumulative distribution function associated with $p(x)$. The resulting $x$ should be distributed according to $p(x)$. Thus to produce a sample $t_{h}^{(i)}$ we classically compute
\begin{equation}
    t_{h}^{(i)} = \frac{1}{\Gamma}\ln\left(\frac{1}{1 - \eta^{(i)}}\right) \text{ where } \eta^{(i)} \sim \mathrm{uniform}(0, 1)
\end{equation}

The quantum circuits associated with a sampled trajectory are depicted below in FIG~\ref{fig:circuit}. Blocks $K(t_{h}^{(i)})$ denote a circuit approximating the unitary evolution generated by $H$ over time $t_{h}^{(i)}$ and blocks $L$ denote a circuit approximating the jump evolution modeled by the quantum channel $\mathcal{J}$. The quantum channels that $K(t_{h}^{(i)})$ and $L$ implement are given by $\mathcal{K}_{t_{h}^{(i)}}$ and $\mathcal{J}'$ respectively.
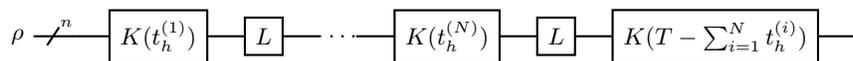
\begin{figure}[h]
    \centering
    \begin{quantikz}
        \lstick{$\rho$} & \qwbundle{n} & \gate{K(t_{h}^{(1)})} & 
        \gate{L} &  \ \ldots & \gate{K(t_{h}^{(N)})} & \gate{L} & \gate{K(T - \sum_{i=1}^{N}t_{h}^{(i)} )} &
    \end{quantikz}
    \caption{Trajectory circuit associated with $\mathcal{L}$}
    \label{fig:circuit}
\end{figure}

To implement each $K(t_{h}^{(i)})$ we can simply choose to use any of the Hamiltonian simulation algorithms available such as \cite{berry2015simulating, berry2015simulating, berry2014exponential, campbell2019random, childs2019fasterquantum, low2019hamiltonian, nakaji2024high}. Since we want to optimize the number of queries to our input $U(H)$, we choose to implement the Hamiltonian simulation components $K(t_{h}^{(i)})$ using the algorithm, first introduced by Low et al. in \cite{low2019hamiltonian} and further analyzed by Gily\'en et al. in \cite{gilyen2019quantum}.

For the circuit $L$, we begin by recalling our average jump $\mathcal{J}$ found in equation~\eqref{eq:jumpevo} and re-stated below,
\begin{align}
    \mathcal{J}(\ket{\psi}\bra{\psi}) = \sum_{\mu=1}^{m}\mathrm{Pr}(\mu|\ket{\psi})\mathcal{J}_{\mu}(\ket{\psi}\bra{\psi}) = \sum_{\mu=1}^{m}\mathrm{Pr}(\mu|\ket{\psi})J_{\mu}(\ket{\psi})J_{\mu}(\ket{\psi})^{\dagger} = \sum_{\mu=1}^{m}\frac{L_{\mu}\ket{\psi}\bra{\psi}L_{\mu}^{\dagger}}{\Gamma}.
\end{align}
Defining Kraus operators $K_{\mu} = L_{\mu} / \sqrt{\Gamma}$, one can readily confirm from applying equation~\eqref{eq:algconstraint} that $K_{\mu}$ satisfy $\sum_{\mu=1}^{m}K_{\mu}^{\dagger} K_{\mu} = \mathbb{I}$, thus $\mathcal{J}$ is a valid quantum channel. This implies that the jump map $\mathcal{J}$ can be approximated faultily through the use of the circuit depicted in FIG~\ref{fig:cpbc} assuming access to $(s_{\mu}, n', \epsilon)$-block encodings for each of the Kraus operators $K_{\mu}$ and oracle $B$ as defined in Lemma~\ref{lm:cpbc}. Since $U(L_{\mu})$ is an $(\alpha_{\mu}, a, 0)$-block encoding of $L_{\mu}$, it is also an $(\alpha_{\mu}/\sqrt{\Gamma}, a, 0)$-block encoding of the Kraus operator $K_{\mu}$ for the quantum channel $\mathcal{J}$. Thus, Lemma~\ref{lm:cpbc} implies the circuit construction of FIG~\ref{fig:cpbc} implements the map $\mathcal{J}$ exactly, with some probability of success $p_{0}$ conditioned on the measurement of the block encoding register in state $\ket{0}$. 

To fully realize the circuit $L$ we need a non-faulty implementation of the map $\mathcal{J}$. Given the faulty implementation of the map $\mathcal{J}$ using the circuit construction described above, we can construct a non-faulty implementation of $\mathcal{J}$ using the technique of oblivious amplification of isometries described in Lemma~\ref{lm:oaa}. Thus we can set
\begin{equation}
    \label{eq:jumpgadget}
    L = [-W(\mathbb{I} - 2 P_{1})W^{\dagger}(\mathbb{I} - 2P_{0})]^{t}W
\end{equation}
where $W$ is the circuit construction described above and depicted in FIG~\ref{fig:cpbc}, with Lemma~\ref{lm:oaa} implying we set $P_{1} = B\ket{0}\bra{0}B^{\dagger} \otimes \mathbb{I} \otimes \mathbb{I}$, and $P_{0} = \mathbb{I}\otimes \ket{0}\bra{0}^{\otimes a} \otimes \mathbb{I}$. The number of times we apply the iterate $- W (\mathbb{I} - 2 P_{1})W^{\dagger}(\mathbb{I} -2 P_{0})$ is the number of times needed to boost the success probability of applying the map $\mathcal{J}$ to one. Lemma~\ref{lm:oaa} implies for an initial success probability of $p_{0}$, we must set $t = O(1 /\sqrt{p_{0}})$ to boost the success probability to one. We can quickly calculate the cost for implementing the circuit $L$ summarized in the Lemma below.
\begin{lemma}
    \label{lm:jumpgadget}
    The quantum circuit described by equation~\eqref{eq:jumpgadget} requires:
    \begin{equation}
        \label{eq:jumpQ}
        O\left(\sqrt{\frac{\sum_{\mu=1}^{m}\alpha_{\mu}^{2}}{\Gamma}}\right)
    \end{equation}
    queries to each $U(L_{\mu})$ and $U(L_{\mu})^{\dagger}$ along with
    \begin{equation}
        O\left((\log(m) + a)\sqrt{\frac{\sum_{\mu=1}^{m}\alpha_{\mu}^{2}}{\Gamma}}\right)
    \end{equation}
    additional 1- and 2-qubit gates, and 
    \begin{equation}
       O\left(\log(m) + a\right)
    \end{equation}
    ancillary qubits.
\end{lemma}
\begin{proof}
    Consider the quantum circuit described by equation~\eqref{eq:jumpgadget}. The ancillary qubit cost needed to implement the full circuit is equal to the number of ancillary qubits needed to implement one application of $W$ since we are allowed to reset the ancillary qubit register in-between applications of $W$. Each application of $W$ requires $O(a)$ ancillary qubits for the block-encoding register along with an additional $O(\log(m))$ ancillary qubits needed to implement the isometry encoding our completely positive map $\mathcal{J}$. Thus the total number of ancillary qubits needed to implement the circuit described by equation~\eqref{eq:jumpgadget} is given by 
    \begin{equation}
        O(\log(m) + a).
    \end{equation}

    As described previously, Lemma~\ref{lm:oaa} implies we must set $t = O(1 /\sqrt{p_{0}})$ in equation~\eqref{eq:jumpgadget} in-order to boost the success probability of implementing map $\mathcal{J}$ to one, where $p_{0}$ denotes the probability of successfully implementing the map $\mathcal{J}$ using only one application of the circuit $W$. Thus $p_{0}$ is the probability of measuring the block-encoding register in state $\ket{0}$ after applying $W$. Calculating $p_{0}$ we find
    \begin{align}
        p_{0} &= \mathrm{Tr}\left([\mathbb{I}_{M}\otimes\ket{0}\bra{0}_{A}\otimes\mathbb{I}] W(\ket{0}\bra{0}_{M}\otimes\ket{0}\bra{0}_{A}\otimes \rho)W^{\dagger}\right) \nonumber \\ 
              &=  \frac{1}{\sum_{\mu=1}^{m}s_{\mu}^{2}}\sum_{\mu,\nu = 1}^{m}s_{\mu} s_{\nu} \mathrm{Tr}\left(\left[\mathbb{I}_{M}\otimes\ket{0}\bra{0}_{A}\otimes\mathbb{I}\right]\ket{\mu}\bra{\nu} \otimes U(L_{\mu})(\ket{0}\bra{0} \otimes \rho)U(L_{\nu})^{\dagger}\right) \nonumber  \\
              &=  \frac{\Gamma}{\sum_{\mu=1}^{m}\alpha_{\mu}^{2}}\sum_{\mu,\nu = 1}^{m}\frac{\alpha_{\mu}\alpha_{\nu}}{\Gamma}\mathrm{Tr}\left(\left[\mathbb{I}_{M}\otimes\ket{0}\bra{0}_{A}\otimes\mathbb{I}\right]\ket{\mu}\bra{\nu} \otimes U(L_{\mu})(\ket{0}\bra{0} \otimes \rho)U(L_{\nu})^{\dagger}\right) \nonumber  \\
              &=  \frac{\Gamma}{\sum_{\mu=1}^{m}\alpha_{\mu}^{2}}\sum_{\mu,\nu = 1}^{m}\frac{1}{\Gamma}\mathrm{Tr}\left(\ket{\mu}\bra{\nu}_{M} \otimes \ket{0}\bra{0}_{A} \otimes L_{\mu}\rho L_{\nu}^{\dagger}\right) =  \frac{\Gamma}{\sum_{\mu=1}^{m}\alpha_{\mu}^{2}}\sum_{\mu=1}^{m}\mathrm{Tr}\left(K_{\mu}\rho K_{\mu}^{\dagger}\right) \nonumber  \\
              &= \frac{\Gamma}{\sum_{\mu=1}^{m}\alpha_{\mu}^{2}}.
    \end{align}
    Thus the circuit described in equation~\eqref{eq:jumpgadget} requires $O\left(\sqrt{(\sum_{\mu=1}^{m}\alpha_{\mu}^{2})/\Gamma}\right)$ queries to the unitaries $W$, $W^{\dagger}$, $\mathbb{I} - 2P_{0}$, and $\mathbb{I} -2P_{1}$.  

    The circuits $W$ and $W^{\dagger}$ require a single query to the unitary $\sum_{\mu=1}^{m}\ket{\mu}\bra{\mu}\otimes U(L_{\mu})$ or $\sum_{\mu=1}^{m}\ket{\mu}\bra{\mu}\otimes U(L_{\mu})^{\dagger}$ respectively. Each $\sum_{\mu=1}^{m}\ket{\mu}\bra{\mu}\otimes U(L_{\mu})$ requires $O(1)$ queries to each $U(L_{\mu})$. Similarly each $\sum_{\mu=1}^{m}\ket{\mu}\bra{\mu}\otimes U(L_{\mu})^{\dagger}$ requires $O(1)$ queries to each $U(L_{\mu})^{\dagger}$. Thus we find the total number of queries to each $U(L_{\mu})$ being
\begin{equation}
    O\left(\sqrt{\frac{\sum_{\mu=1}^{m}\alpha_{\mu}^{2}}{\Gamma}}\right),
\end{equation}
and similar for each $U(L_{\mu})^{\dagger}$.

The circuits $W$, $W^{\dagger}$, and $\mathbb{I} - 2P_{1}$ all require $O(1)$ queries to the oracles $B$ and $B^{\dagger}$. Since the only constraint on $B$ is that it maps $\ket{0}$ to $\ket{\omega} = \sqrt{\Gamma / (\sum_{\mu=1}^{m}\alpha_{\mu}^{2})}\sum_{\mu=1}^{m}(\alpha_{\mu} / \sqrt{\Gamma})\ket{\mu}$, each $B$ and $B^{\dagger}$ can be implemented in $O(\log(m))$ 1- and 2-qubit gates. Thus the contribution to the 1- and 2-qubit gate count arising from queries to $B$ is
\begin{equation}
    \label{eq:b}
    O\left(\log(m)\sqrt{\frac{\sum_{\mu=1}^{m}\alpha_{\mu}^{2}}{\Gamma}}\right).
\end{equation}

Considering the reflections $\mathbb{I} - 2P_{0}$ and $\mathbb{I} - 2\ket{0}\bra{0} \otimes \mathbb{I} \otimes \mathbb{I}$, each can be implemented with $O(a)$ and  $O(\log(m))$ 1- and 2-qubit gates respectively. Thus the contribution to the 1- and 2-qubit gate count arising from queries to the reflections $\mathbb{I} - 2P_{0}$ and $\mathbb{I} - 2\ket{0}\bra{0} \otimes \mathbb{I} \otimes \mathbb{I}$ is 
\begin{equation}
    \label{eq:refl}
    O\left((\log(m) + a)\sqrt{\frac{\sum_{\mu=1}^{m}\alpha_{\mu}^{2}}{\Gamma}}\right).
\end{equation}

Equations~\eqref{eq:b} and~\eqref{eq:refl} imply the total additional 1- and 2-qubit gate count to implement the circuit described in equation~\eqref{eq:jumpgadget} is
\begin{equation}
    O\left((\log(m) + a)\sqrt{\frac{\sum_{\mu=1}^{m}\alpha_{\mu}^{2}}{\Gamma}}\right)
\end{equation}
\end{proof}

With the randomized compilation procedure described above, along with constructions for components $K(t_{h}^{(i)})$ and $L$, we proceed to analyze the query and 1- and 2-qubit gate costs our quantum algorithm requires in the worst-case.
\subsection{Analysis}
To analyze the costs of the algorithm described above we must first bound the error between the quantum channel generated by the Lindbladian dynamics and the quantum channel our algorithm implements, denoted by $\mathcal{E}_{T,r}$. This bound on our simulation error, given in Lemma~\ref{lm:algerror}, allows us to fix a value for $r$, the maximum number of jumps allowed to be compiled, which assures the total simulation error is of $O(\epsilon)$. To compute the final resource estimates for implementing the algorithm, we analyze the worst-case circuit our randomized compilation procedure produces for the appropriate value of $r$ and $T$.

Before proving Lemma~\ref{lm:algerror} we must first characterize the map $\mathcal{E}_{T,r}$ our algorithm implements. For a fixed simulation time $T$ and maximum number of compilable jumps $r$, denote the probability density of compiling a circuit implementing the map $\mathcal{C}_{t_{h}^{(1)}, \cdots, t_{h}^{(N)}}$ with holding times $\{t_{h}^{(i)}\}_{i=1}^{N}$ as $f(t_{h}^{(1)},\cdots, t_{h}^{(N)}; T, r)$. Thus the map our algorithm implements is defined as
\begin{equation}
    \label{eq:algchannel}
    \mathcal{E}_{T,r} = \mathbb{E}_{f(T,r)}\left[\mathcal{C}_{t_{h}^{(1)}, \cdots, t_{h}^{(N)}}\right] =  \sum_{N = 1}^{\infty}\int_{0}^{T}dt_{h}^{(1)} \cdots \int_{0}^{T} dt_{h}^{(N)}f(t_{h}^{(1)},\cdots, t_{h}^{(N)}; T, r)\mathcal{C}_{t_{h}^{(1)}, \cdots, t_{h}^{(N)}}.
\end{equation}
All that is left is to characterize the probability density $f(T, r)$ which our randomized compilation procedure implements.

Consider the randomized compilation procedure used in our algorithm. The protocol amounts to sampling a set of holding times $\{t_{h}^{(i)}\}_{i = 1}^{N} \sim p(T)$ and rejecting the sample if $N > r$. This procedure is repeated until a sample is accepted and the circuit implementing the map $\mathcal{C}_{t_{h}^{(1)}, \cdots, t_{h}^{(N)}}$ is compiled. Given that we reject samples with $N > r$ we can write down the probability density for $f(T, r)$ as
\begin{equation}
    \label{eq:algpdensity}
    f(t_{h}^{(1)}, \cdots t_{h}^{(N)}; T, r) = \begin{cases} p(t_{h}^{(1)}, \cdots, t_{h}^{(N)}; T) \mathrm{Pr}(N(T) \leq  r) \text{ for } N \leq r \\ 0 \text{ otherwise } \end{cases},
\end{equation}
where the true probability density for sampling a trajectory $p(t_{h}^{(1)}, \cdots, t_{h}^{(N)}; T)$ is defined in equation~\eqref{eq:pdensity}. Thus with $f(T, r)$ fully described we can proceed to proving Lemma~\ref{lm:algerror} which bounds the error between the simulation and true dynamics in diamond distance.
\begin{lemma}
    \label{lm:algerror}
    Let $\mathcal{E}_{T,r}$ denote the map our algorithm implements defined in equation~\eqref{eq:algchannel}. Then
    \begin{equation}
        d_{\diamond}(e^{T\mathcal{L}}, \mathcal{E}_{T,r}) \leq \left(\frac{e \Gamma T}{r}\right)^{r}e^{-\Gamma T} + (r + 1) \epsilon_{H}
    \end{equation}
    where $\epsilon_{H} = \sup\left(\left\{d_{\diamond}\left(\mathcal{U}_{t_{h}^{(i)}}, \mathcal{K}_{t_{h}^{(i)}}\right)\right\}_{i = 1}^{N} \cup \left\{d_{\diamond}\left(\mathcal{U}_{T - \sum_{i=1}^{N}t_{h}^{(i)}}, \mathcal{K}_{T - \sum_{i=1}^{N}t_{h}^{(i)}}\right)\right\}\right)$.
\end{lemma}
\begin{proof}
    The goal is to upperbound $d_{\diamond}(e^{T\mathcal{L}}, \mathcal{E}_{T,r})$. From equations~\eqref{eq:constraintqtraj}, and~\eqref{eq:algchannel} it follows that
    \begin{align}
        d_{\diamond}(e^{T\mathcal{L}}, \mathcal{E}_{T,r}) &= \frac{1}{2}\dnorm{e^{T\mathcal{L}} - \mathcal{E}_{T,r}} = \frac{1}{2}\dnorm{\mathbb{E}_{p(T)}\left[\mathcal{T}_{t_{h}^{(1)},\cdots, t_{h}^{(N)}}\right] - \mathbb{E}_{f(T,r)}\left[\mathcal{C}_{t_{h}^{(1)},\cdots, t_{h}^{(N)} }\right]} \nonumber \\
                                                    &\leq \frac{1}{2}\sum_{N=0}^{\infty}\int_{0}^{T}dt_{h}^{(1)} \cdots \int_{0}^{T}dt_{h}^{(N)}\dnorm{p(t_{h}^{(1)}, \cdots, t_{h}^{(N)} ; T)\mathcal{T}_{t_{h}^{(1)},\cdots, t_{h}^{(N)} } - f(t_{h}^{(1)},\cdots, t_{h}^{(N)} ; T, r) \mathcal{C}_{t_{h}^{(1)},\cdots, t_{h}^{(N)} }} \nonumber \\
                                                    &\leq \frac{1}{2}\sum_{N=0}^{\infty}\int_{0}^{T}dt_{h}^{(1)} \cdots \int_{0}^{T}dt_{h}^{(N)} \left(|p(t_{h}^{(1)}, \cdots, t_{h}^{(N)} ; T ) - f(t_{h}^{(1)}, \cdots, t_{h}^{(N)}; T, r )|\dnorm{\mathcal{T}_{t_{h}^{(1)},\cdots, t_{h}^{(N)}}} \right. \nonumber \\ &+ \left. f(t_{h}^{(1)},\cdots, t_{h}^{(N)}; T, r) \dnorm{\mathcal{T}_{t_{h}^{(1)}, \cdots, t_{h}^{(N)}} - \mathcal{C}_{t_{h}^{(1)}, \cdots, t_{h}^{(N)}}}\right)
    \end{align}
where the inequalities follow from the sub-additivity and sub-multiplicativity of the diamond norm. Given that $\mathcal{T}_{t_{h}^{(1)},\cdots, t_{h}^{(N)}}$ is a completely positive and trace preserving map, $\dnorm{\mathcal{T}_{t_{h}^{(1)},\cdots, t_{h}^{(N)} }} = 1$, allowing us to write
\begin{align}
    d_{\diamond}(e^{T\mathcal{L}}, \mathcal{E}_{T,r}) &\leq \frac{1}{2}\sum_{N=0}^{\infty}\int_{0}^{T}dt_{h}^{(1)} \cdots \int_{0}^{T}dt_{h}^{(N)} \left(|p(t_{h}^{(1)},\cdots, t_{h}^{(N)}; T) - f(t_{h}^{(1)},\cdots, t_{h}^{(N)}; T, r)| \right. \nonumber \\ &+ \left. f(t_{h}^{(1)}, \cdots, t_{h}^{(N)}; T, r) \dnorm{\mathcal{T}_{t_{h}^{(1)},\cdots, t_{h}^{(N)}} - \mathcal{C}_{t_{h}^{(1)},\cdots, t_{h}^{(N)}}}\right).
\end{align}
The first term in parenthesis is just the total variational distance between the probability densities $p(T)$ and $f(T, r)$ while the second term is just the error of a trajectory averaged over $f(T, r)$, thus we can bound the error of our simulation as:
\begin{equation}
    \label{eq:qtrajerror}
    d_{\diamond}(e^{T\mathcal{L}},\mathcal{E}_{T, r}) \leq \delta(p(T), f(T,r)) + \frac{1}{2}\mathbb{E}_{f(T,r)}\left[\dnorm{\mathcal{T}_{t_{h}^{(1)},\cdots, t_{h}^{(N)}} - \mathcal{C}_{t_{h}^{(1)},\cdots, t_{h}^{(N)}}}\right].
\end{equation}

Focusing first on the second term in equation~\eqref{eq:qtrajerror} we find 
\begin{align}
    \label{eq:hsimerror}
    \dnorm{\mathcal{T}_{t_{h}^{(1)},\cdots, t_{h}^{(N)}} - \mathcal{C}_{t_{h}^{(1)},\cdots, t_{h}^{(N)}}} &= \dnorm{\mathcal{U}_{(T-\sum_{i=1}^{N}t_{h}^{(i)})} \circ \mathcal{J} \circ \mathcal{U}_{t_{h}^{(N)}} \cdots \mathcal{J} \circ \mathcal{U}_{t_{h}^{(1)}} - \mathcal{K}_{(T - \sum_{i=1}^{N}t_{h}^{(N)})} \circ \mathcal{J} \circ \mathcal{K}_{t_{h}^{(N)}} \cdots \mathcal{J} \circ \mathcal{K}_{t_{h}^{(1)}}} \nonumber \\
                                                                                                          &\leq \dnorm{\mathcal{U}_{(T - \sum_{i=1}^{N}t_{h}^{(i)})} - \mathcal{K}_{(T - \sum_{i=1}^{N}t_{h}^{(i)})}} + \sum_{i=1}^{N}\dnorm{\mathcal{U}_{t_{h}^{(i)}} - \mathcal{K}_{t_{h}^{(i)}}}
\end{align}
which is just the sum over all the errors incurred by each segment of Hamiltonian simulation. Define 
\begin{equation}
    \epsilon_{H} = \sup\left(\left\{d_{\diamond}\left(\mathcal{U}_{t_{h}^{(i)}}, \mathcal{K}_{t_{h}^{(i)}}\right)\right\}_{i = 1}^{N} \cup \left\{d_{\diamond}\left(\mathcal{U}_{T - \sum_{i=1}^{N}t_{h}^{(i)}}, \mathcal{K}_{T - \sum_{i=1}^{N}t_{h}^{(i)}}\right)\right\}\right),
\end{equation}
which represents the largest error incurred by a Hamiltonian simulation segment. Thus equation~\eqref{eq:hsimerror} becomes
\begin{align}
    \dnorm{\mathcal{T}_{t_{h}^{(1)},\cdots, t_{h}^{(N)}} - \mathcal{C}_{t_{h}^{(1)},\cdots, t_{h}^{(N)}}} \leq (N+1)\epsilon_{H}
\end{align}
which implies 
\begin{equation}
    \mathbb{E}_{f(T, r)}\left[\dnorm{\mathcal{T}_{t_{h}^{(1)},\cdots, t_{h}^{(N)}} - \mathcal{C}_{t_{h}^{(1)},\cdots, t_{h}^{(N)}}}\right] \leq \sum_{N=0}^{\infty}\int_{0}^{T}dt_{h}^{(1)} \cdots \int_{0}^{T}dt_{h}^{(N)} f(t_{h}^{(1)}, \cdots, t_{h}^{(N)}; T, r) (N+1)\epsilon_{H}.
\end{equation}
Given that our randomized compilation procedure modeled by $f(T, r)$ never compiles a circuit with more than the maximum of $r$ quantum jumps, equation~\eqref{eq:algpdensity} implies $f(t_{h}^{(1)}, \cdots, t_{h}^{(N)}; T, r)=0$ for $N > r$. Thus we find
\begin{equation}
    \label{eq:circuiterrorupperbound}
    \mathbb{E}_{f(T, r)}\left[\dnorm{\mathcal{T}_{t_{h}^{(1)},\cdots, t_{h}^{(N)}} - \mathcal{C}_{t_{h}^{(1)},\cdots, t_{h}^{(N)}}}\right] \leq (r+1)\epsilon_{H}.
\end{equation}

Next, we focus on upperbounding the total variational distance between $p(T)$ and $f(T,r)$. Using equations~\eqref{eq:pdensity} and~\eqref{eq:algpdensity} we calculate
\begin{align}
    \delta(p(T),f(T,r)) &= \frac{1}{2}\sum_{N=0}^{\infty}\int_{0}^{T}dt_{h}^{(1)} \cdots \int_{0}^{T}dt_{h}^{(N)} |p(t_{h}^{(1)},\cdots, t_{h}^{(N)}; T) - f(t_{h}^{(1)},\cdots, t_{h}^{(N)}; T, r)| \nonumber \\ 
                &= \frac{1}{2}\left(\sum_{N=0}^{r}\int_{0}^{T}dt_{h}^{(1)} \cdots \int_{0}^{T}dt_{h}^{(N)} p(t_{h}^{(1)},\cdots, t_{h}^{(N)}; T)|1  - \mathrm{Pr}(N(T) \leq r)| \right.\nonumber\\
                &\left.+ \sum_{N=r+1}^{\infty}\int_{0}^{T}dt_{h}^{(1)} \cdots \int_{0}^{T}dt_{h}^{(N)} p(t_{h}^{(1)},\cdots, t_{h}^{(N)}; T)\right),
\end{align}
which gives us
\begin{align}
    \delta(p(T),f(T,r)) &\leq \frac{1}{2}\left(\sum_{N=0}^{r}\int_{0}^{T}dt_{h}^{(1)} \cdots \int_{0}^{T}dt_{h}^{(N)} \mathrm{Pr}\left(t_{h}^{(1)},\cdots, t_{h}^{(N)}\big| S^{(N)} \leq T, S^{(N+1)} > T\right) \right. \nonumber \\ 
                & \times \mathrm{Pr}(N(T) = N)|1  - \mathrm{Pr}(N(T) \leq r)| \nonumber\\
                &\left.+ \sum_{N=r+1}^{\infty}\int_{0}^{T}dt_{h}^{(1)} \cdots \int_{0}^{T}dt_{h}^{(N)} \mathrm{Pr}\left(t_{h}^{(1)},\cdots, t_{h}^{(N)} \big| S^{(N)} \leq T, S^{(N+1)} > T \right)\mathrm{Pr}(N(T) = N)\right) \nonumber \\
                &\leq \frac{1}{2}\left(\sum_{N=0}^{r}\mathrm{Pr}(N(T) = N)|1  - \mathrm{Pr}(N(T) \leq r)| + \sum_{N=r+1}^{\infty}\mathrm{Pr}(N(T)=N) \right) \nonumber \\
                &\leq \frac{1}{2}\left(\mathrm{Pr}(N(T) \leq r)|1  - \mathrm{Pr}(N(T) \leq r)|  + \mathrm{Pr}(N(T) > r ) \right) = \frac{1}{2}\left(\mathrm{Pr}(N(T) \leq r) + 1 \right)\mathrm{Pr}(N(T)>r)\nonumber \\
                &\leq \frac{1}{2}\left(2 - \mathrm{Pr}(N(T) > r)\right)\mathrm{Pr}(N(T)>r) \leq \mathrm{Pr}(N(T)>r).
\end{align}
All that is left is to upperbound the right-tail of the distribution for $N(T)$ which is the random variable representing the number of quantum jumps occurring in an interval of time $T$. Applying Theorem~\ref{thm:pois} of Appendix~\ref{app:B} we find  $N(T) \sim \mathrm{Pois}(\Gamma T)$, thus  we can apply the Chernoff bound for the upper tail of the Poisson distribution, Theorem 5.4 of \cite{Mitzenmacher_Upfal_2005}, to find
\begin{equation}
    \label{eq:tvdupperbound}
    \delta(p(T),f(T, r)) \leq \mathrm{Pr}(N(T) > r) \leq \left(\frac{e\Gamma T}{r}\right)^{r}e^{-\Gamma T}.
\end{equation}

Equations~\eqref{eq:circuiterrorupperbound} and~\eqref{eq:tvdupperbound} allow us to write down the final upperbound on the error of our simulation as
\begin{equation}
    d_{\diamond}(e^{T\mathcal{L}}, \mathcal{E}_{T,r}) \leq \left(\frac{e\Gamma T}{r}\right)^{r}e^{-\Gamma T} + \left(r + 1\right)\epsilon_{H}.
\end{equation}
\end{proof}

Armed with the bound on simulation error provided by Lemma~\ref{lm:algerror} we can proceed to proving the main result stated below.
\begin{theorem}
    \label{thm:mainres}
    Let $\mathcal{L}$ be a Lindbladian with Hamiltonian $H$ and jump operators $\{L_{\mu}\}_{\mu=1}^{m}$. Let the jump operators satisfy $\sum_{\mu = 1}^{m}L_{\mu}^{\dagger}L_{\mu} = \Gamma \mathbb{I}$ for some $\Gamma \in \mathbb{R}_{+}$. Given access to an $(\alpha, a, 0)$-block-encoding of the Hamiltonian $U(H)$ and $(\alpha_{\mu}, a, 0)$-block-encodings for each of the jump operator $U(L_{\mu})$ as well as their inverses, then for all $T \in \mathbb{R}_{+} \setminus \{0\}$ and $\epsilon \in (0, 1)$ there exists a quantum algorithm which approximates the action of the channel $e^{T\mathcal{L}}$ up to precision $\epsilon$ in diamond distance requiring in the worst-case
    \begin{equation}
        O\left(\sqrt{\frac{\sum_{\mu=1}^{m}\alpha_{\mu}^{2}}{\Gamma}}\left(\Gamma T + \frac{\log(1 / \epsilon)}{\log(e + \frac{\log(1 / \epsilon)}{\Gamma T})}\right)\right)
    \end{equation}
    queries to each $U(L_{\mu})$ and $U(L_{\mu})^{\dagger}$ and
    \begin{equation}
        O\left((\alpha + \Gamma)T\left[\frac{\log((\alpha + \Gamma) T / \epsilon)}{\log(e + \frac{\log((\alpha + \Gamma) T / \epsilon)}{(\alpha + \Gamma)T})}\right]^{2} \right)
    \end{equation}
    queries to $U(H)$ and $U(H)^{\dagger}$ along with
    \begin{equation}
        O\left((\log(m) + a)\sqrt{\frac{\sum_{\mu=1}^{m}\alpha_{\mu}^{2}}{\Gamma}}(\alpha + \Gamma)T\left[\frac{\log((\alpha + \Gamma) T / \epsilon)}{\log(e + \frac{\log((\alpha + \Gamma) T / \epsilon)}{(\alpha + \Gamma)T})}\right]^{2} \right)
    \end{equation}
    additional 1- and 2-qubit gates and
    \begin{equation}
        O(\log(m) + a)
    \end{equation}
    ancillary qubits.
\end{theorem}
\begin{proof}
    Consider the quantum algorithm described in the section above. Lemma~\ref{lm:algerror} implies that 
    \begin{equation}
        d_{\diamond}(e^{T\mathcal{L}}, \mathcal{E}_{T,r}) \leq \left(\frac{e\Gamma T}{r}\right)^{r} + (r + 1)\epsilon_{H}
    \end{equation}
    where $\mathcal{E}_{T,r}$ denotes the quantum channel our algorithm implements for a fixed simulation time $T$ and maximum number of quantum jumps compiled $r$. Requiring the above simulation error to be upperbounded by $\epsilon$ implies it suffices to constrain
    \begin{equation}
        \label{eq:upalgerror}
        \left(\frac{e\Gamma T}{r}\right)^{r} \leq \frac{\epsilon}{2} \quad \text{ and } \quad (r + 1) \epsilon_{H} \leq \frac{\epsilon}{2}.
    \end{equation}
    Lemma 59 of \cite{gilyen2019quantum} implies the first of the two inequalities above is satisfied by taking
    \begin{equation}
        r = O\left(\Gamma T + \frac{\log(1 / \epsilon)}{\log(e + \frac{\log(1 / \epsilon)}{(\Gamma T)})}\right)
    \end{equation}

    The worst-case quantum circuit built using our randomized compiling procedure requires a number of quantum jumps equal to $r$. Thus, in the worst-case we require $r$ queries to the circuit approximating the quantum jumps, denoted $L$ in FIG~\ref{fig:circuit}, along with Hamiltonian simulation for the sequence of times $\{t^{(i)}_{h}\}_{i = 1}^{r}$. In-order to implement the $r$ quantum jump gadgets in a sampled circuit, Lemma~\ref{lm:jumpgadget} implies that in the worst-case the circuit requires
    \begin{equation}
        \label{eq:Lqueries}
        O\left(r\sqrt{\frac{\sum_{\mu=1}^{m}\alpha_{\mu}^{2}}{\Gamma}}\right) = O\left(\sqrt{\frac{\sum_{\mu=1}^{m}\alpha_{\mu}^{2}}{\Gamma}}\left(\Gamma T + \frac{\log(1 / \epsilon)}{\log(e + \frac{\log(1 / \epsilon)}{(\Gamma T)})}\right)\right)
    \end{equation}
    queries to each $U(L_{\mu})$ and $U(L_{\mu})^{\dagger}$ along with 
    \begin{align}
        \label{eq:jumpgates}
        O\left(r(\log(m) + a)\sqrt{\frac{\sum_{\mu=1}^{m}\alpha_{\mu}^{2}}{\Gamma}}\right) = \nonumber \\
        &O\left((\log(m) + a)\sqrt{\frac{\sum_{\mu=1}^{m}\alpha_{\mu}^{2}}{\Gamma}}\left(\Gamma T + \frac{\log(1 / \epsilon)}{\log(e + \frac{\log(1 / \epsilon)}{(\Gamma T)})}\right)\right)
    \end{align}
    1- and 2-qubit gates and
    \begin{equation}
        \label{eq:jumpancilla}
        O(\log(m) + a)
    \end{equation}
    ancillary qubits assuming we can reset the ancillary qubits between quantum jump gadgets.

    Considering the second inequality in equation~\eqref{eq:upalgerror}, we find that the constraint holds if we take $\epsilon_{H} \leq \epsilon / 2(r + 1)$. Given that $\epsilon_{H}$ is defined as the largest error incurred by one of our Hamiltonian simulation maps $\mathcal{K}_{t_{h}^{(i)}}$, it suffices to take the precision for each our Hamiltonian simulation subroutines to be $\epsilon_{H} = O(\epsilon /r )$. This assures that the error incurred by the $r$ Hamiltonian simulation components of a randomly compiled circuit is of $O(\epsilon)$. Since we implement each of the $r$ Hamiltonian simulation components using the Hamiltonian simulation algorithm of Low et al. \cite{low2019hamiltonian}, we can leverage Theorem 58 of \cite{gilyen2019quantum} to determine the cost of implementing a single component for holding time $t_{h}^{(i)}$ and precision $\epsilon_{H}$ finding that it requires:
    \begin{equation}
        O\left(\alpha t_{h}^{(i)} + \frac{\log(1/\epsilon_{H})}{\log(e + \frac{\log(1 / \epsilon_{H})}{(\alpha t_{h}^{(i)})})}\right)
    \end{equation}
    queries to $U(H)$ and $U(H)^{\dagger}$, along with
    \begin{equation}
        O\left(a\left(\alpha t_{h}^{(i)} + \frac{\log(1/\epsilon_{H})}{\log(e + \frac{\log(1 / \epsilon_{H})}{(\alpha t_{h}^{(i)})})}\right)\right)
    \end{equation}
    1- and 2-qubit gates, and $O(1)$ 
    ancillary qubits. Thus the total number of queries to $U(H)$ required to implement the $r$ Hamiltonian simulation components in our worst-case circuit is
    \begin{align}
        \label{eq:Hqueries}
        O\left(\sum_{i=1}^{r}\left(\alpha t_{h}^{(i)} + \frac{\log(r/\epsilon)}{\log(e + \frac{\log(r / \epsilon)}{(\alpha t_{h}^{(i)})})}\right)
        \right)&= O\left(\alpha T + r\frac{\log(r / \epsilon)}{\log(e + \frac{\log(r / \epsilon)}{\alpha T})}\right) \nonumber \\
               &= O\left(\alpha T + r\frac{\log(\Gamma T /\epsilon^{2})}{\log(e + \frac{\log(\Gamma T / \epsilon^{2})}{\alpha T})}\right) \nonumber \\
               &= O\left(\alpha T + \left(\Gamma T + \frac{\log(1 / \epsilon)}{\log(e + \frac{\log(1/\epsilon)}{\Gamma T})}\right)\frac{\log(\Gamma T / \epsilon)}{\log(e + \frac{\log(\Gamma T / \epsilon)}{\alpha T})}\right) \nonumber \\
               &= O\left(\alpha T + \Gamma T \left[\frac{\log(\Gamma T / \epsilon)}{\log(e + \frac{\log(\Gamma T / \epsilon)}{(\alpha + \Gamma)T})}\right]^{2} \right) \nonumber \\
               & = O\left((\alpha + \Gamma)T\left[\frac{\log((\alpha + \Gamma) T / \epsilon)}{\log(e + \frac{\log((\alpha + \Gamma) T / \epsilon)}{(\alpha + \Gamma)T})}\right]^{2} \right),
    \end{align}
    where we have used the fact that $\log(x)/\log(e + \frac{\log(x)}{\gamma})$ is monotonic increasing on $x \in \mathbb{R}_{+}\setminus \{0\}$ for any $\gamma \in \mathbb{R}_{+} \setminus \{0\}$. Similarly the total number of 1- and 2-qubit gates needed to implement the $r$ Hamiltonian simulation components follows as
    \begin{align}
        \label{eq:Hgates}
        O\left(a\sum_{i=1}^{r}\left(\alpha t_{h}^{(i)} + \frac{\log(r/\epsilon)}{\log(e + \frac{\log(r / \epsilon)}{(\alpha t_{h}^{(i)})})}\right)
        \right)&= O\left(a(\alpha + \Gamma)T\left[\frac{\log((\alpha + \Gamma) T / \epsilon)}{\log(e + \frac{\log((\alpha + \Gamma) T / \epsilon)}{(\alpha + \Gamma)T})}\right]^{2} \right).
    \end{align}
    Finally, assuming we can reset the ancillary qubits between Hamiltonian simulation components implies our ancillary qubit cost to implement the $r$ Hamiltonian simulation components is still $O(1)$.

    Combining equations~\eqref{eq:jumpgates}, and \eqref{eq:Hgates}, gives us the total additional 1- and 2-qubit gate cost for implementing the worst case circuit as:
    \begin{equation}
        O\left((\log(m) + a)\sqrt{\frac{\sum_{\mu=1}^{m}\alpha_{\mu}^{2}}{\Gamma}}(\alpha + \Gamma)T\left[\frac{\log((\alpha + \Gamma) T / \epsilon)}{\log(e + \frac{\log((\alpha + \Gamma) T / \epsilon)}{(\alpha + \Gamma)T})}\right]^{2} \right).
    \end{equation}
    And the total ancillary qubit cost for implementing the worst case circuit as:
    \begin{equation}
        O(\log(m) + a).
    \end{equation}
    Finally equations~\eqref{eq:Hqueries} and~\eqref{eq:Lqueries} give us the total query costs to $U(H)$, each $U(L_{\mu})$, and their inverses respectively. The proof of the statement follows.
\end{proof}

\section{Characterization of the restricted Lindbladians}
\label{sec:char}
To begin, let us define the set $\algL$ as the set of all Lindbladians with a representation given by a Hamiltonian $H$ and jump operators $\{L_{\mu}\}_{\mu=1}^{m}$ which satisfy
\begin{equation}
    \label{eq:algconstraint:repeated}
    \sum_{\mu=1}^{m}L_{\mu}^{\dagger}L_{\mu} = \Gamma \mathbb{I},
\end{equation}
where $\Gamma \in \mathbb{R}_{+}$. Thus $\algL$ consists of all possible Lindbladians our algorithm can simulate. The lower-bound on the number of queries needed to a block-encoding of a Lindbladian's jump operator for Lindbladians in $\algL$ follows directly from a result by Gao et al. given below.
\begin{theorem}[Theorem D.4 from \cite{gao2026levy}] Let $\mathcal{L}$ denote the following Lindbladian,
    \label{thm:gao}
    \begin{equation}
        \mathcal{L}(\rho) = \sum_{j=1}^{N}p_{j}e^{iHs_{j}}\rho e^{-iHs_{j}} - \rho
    \end{equation}
    where $H$ is a Hamiltonian with $\norm{H} \leq 1$, $p_{j} \geq 0$ and $\sum_{j=1}^{N}p_{j} = 1$ and $s_{j}$ are fixed constants. Let $U(H)$ be an $(1, a, 0)$-block-encoding of $H$, then simulating the evolution of $\mathcal{L}$ for time $T$ requires $\Omega(T)$ queries to $U(H)$, even for the case of $N=1$ where
    \begin{equation}
        \mathcal{L}(\rho) = e^{iH}\rho e^{-iH} - \rho.
    \end{equation}
\end{theorem}

The above result by Gao et al. implies the following corollary.

\begin{corollary}
    There exists an $\mathcal{L} \in \{\mathcal{L}_{c}\}$ with a single jump operator $L$ such that to simulate the evolution of $\mathcal{L}$ for time $T$ requires $\Omega(T)$ queries to a block-encoding of $U(L)$.
\end{corollary}
\begin{proof}
    Choose the Lindbladian $\mathcal{L}$ generated by no Hamiltonian and $L = e^{iH}$ where $H$ satisfies Theorem~\ref{thm:gao}. $L^{\dagger}L = \mathbb{I}$, thus $\mathcal{L} \in \algL$. We require $\Theta(1)$ queries to $U(H)$ in-order construct $U(L)$ as input to a Lindbladian simulation algorithm. Theorem~\ref{thm:gao} implies we require $\Omega(T)$ queries to $U(H)$ in order to simulate $\mathcal{L}$ for time $T$, thus we require $\Omega(T)$ queries to $U(L)$ in-order to simulate $\mathcal{L}$ for time $T$.
\end{proof}

For purely dissipative Lindbladians, Lindbladians without a Hamiltonian, the above statement implies our algorithm achieves the optimal jump operator query complexity for our restricted set of Lindbladians $\algL$.

Given our algorithm presented in the previous section, one possible avenue towards lifting our constraint of equation~\eqref{eq:algconstraint} would be to ``embed" an arbitrary Lindbladian into one which satisfies our algorithm's constraint. Such an ``embedding" would allow one to leverage the above algorithm as a black-box to simulate arbitrary Lindbladians at the cost of the embedding procedure. In what follows we provide  a characterization of our restricted set Lindbladians and various results suggesting that a black-box ``embedding" of an arbitrary Lindbadian into this restricted family is not possible.

It is well known that given a Lindbladian $\mathcal{L}$ the representation of $\mathcal{L}$ given by a Hamiltonian $H$ and jump operators $\{L_{\mu}\}_{\mu=1}^{m}$ is not unique. The Lindbladian superoperator $\mathcal{L}$ is invariant under a unitary transformation of its jump operators given by
\begin{equation}
    \label{eq:ut}
    L_{\mu}' = \sum_{\alpha}u_{\mu\alpha}L_{\alpha}
\end{equation}
where $u_{\mu\alpha}$ are the matrix elements of unitary transformation as well as an inhomogeneous transformation of both its Hamiltonian and jump operators given by
\begin{align}
    \label{eq:it}
    L_{\mu}' &= L_{\mu} + a_{\mu}\mathbb{I} \\
    \label{eq:it2}
    H' &= H + \frac{1}{2i}\sum_{\mu}\left(a_{\mu}^{*}L_{\mu} - a_{\mu}L_{\mu}^{\dagger}\right) + b\mathbb{I},
\end{align}
with $a_{\mu} \in \mathbb{C}$ and $b\in \mathbb{R}$. It is manifest that equation~\eqref{eq:algconstraint} is invariant under transformations given by equation~\eqref{eq:ut} but the inhomogeneous transformation given by equations~\eqref{eq:it} and~\eqref{eq:it2} may break this invariance. Given our task of simulating Lindbladians with a representation under the constraint of equation~\eqref{eq:algconstraint}, one may ask whether any Lindbladian can be represented in-terms of jump operators which satisfy equation~\eqref{eq:algconstraint}. This question is answered in the negative by the following lemma.
\begin{lemma}
    \label{lm:subset}
    There exists a Lindbladian $\mathcal{L} \notin \algL$
\end{lemma}
\begin{proof}

    We proceed by proof by contradiction. Let $\mathcal{L}$ be the Lindbladian which generates the amplitude damping channel, represented by the jump operators $L=\ket{0}\bra{1}$ and Hamiltonian $H=0$. Note, such a representation of $\mathcal{L}$ does not satisfy equation~\eqref{eq:algconstraint}. Our assumption $\mathcal{L} \in \{\mathcal{L}\}_{c}$ implies there exists a transformation of the representation of $\mathcal{L}$ such that equation~\eqref{eq:algconstraint} is satisfied. Equation~\eqref{eq:algconstraint} is invariant under transformations of the form of equation~\eqref{eq:ut} thus without loss of generality we restrict to considering only inhomogeneous transformations given by equations~\eqref{eq:it} and~\eqref{eq:it2}.

    Equations~\eqref{eq:it} and~\eqref{eq:it2} imply under an arbitrary inhomogeneous transformation 
    \begin{align}
        L &\mapsto L' = \ket{0}\bra{1} + a\mathbb{I} \\
        H &\mapsto H' = \frac{1}{2i}\left(a^{*} \ket{0}\bra{1} - a \ket{1}\bra{0}\right) + b \mathbb{I}
    \end{align}
    Let $a$ and $b$ denote the parameters of the inhomogeneous transformation mapping the original representation of the amplitude damping Lindbladian into one which satisfies equation~\eqref{eq:algconstraint}.
    Checking equation~\eqref{eq:algconstraint} we find 
    \begin{equation}
        L'^{\dagger}L' = (1 + |a|^{2})\ket{1}\bra{1} + a\ket{1}\bra{0} + a^{*}\ket{0}\bra{1} + |a|^{2}\ket{0}\bra{0} = \Gamma \mathbb{I}
    \end{equation}
    for some $\Gamma \in \mathbb{R}_{+}$. The above equation implies
    \begin{equation}
        \begin{cases}
            1 + |a|^{2} = \Gamma \\
            |a|^{2} = \Gamma \\
            a = 0
        \end{cases}
    \end{equation}
    which requires $1 + \Gamma = \Gamma$, a contradiction.
\end{proof}

Lemma~\ref{lm:subset} implies $\algL$ is a strict subset of all Lindbladians, indicating our algorithm is not universal in its applicability. Although Lemma~\ref{lm:subset} does imply there exists Lindbladians which our algorithm cannot simulate, the set of $\algL$ may become universal with respect to inducing Lindbladian dynamics on a subsystem. In essence, it may be possible to take an $\mathcal{L} \notin \algL$ and embed its dynamics inside a $\mathcal{K} \in \algL$ which generates dynamics on a larger system. The strongest way to formalize this statement is by requiring for all Lindbladians $\mathcal{L}$ there exists a $\mathcal{K} \in \algL$ such that for all time $t$ and for some state $\omega_{A}$
\begin{equation}
    \Tr_{A} [e^{t\mathcal{K}}(\omega_{A}\otimes \rho)] = e^{t\mathcal{L}}(\rho).
\end{equation}
Such an embedding procedure can be shown to be impossible by the following statement.
\begin{theorem}
    \label{lm:noembed}
    For all $t \in \mathbb{R}_{+}$, let  $\Tr_{A} [e^{t\mathcal{K}}(\omega_{A}\otimes \rho)] = e^{t\mathcal{L}}(\rho)$. If $\mathcal{K} \in \algL$ then $\mathcal{L} \in \algL$.
\end{theorem}
\begin{proof}
    Assume for all $t \in \mathbb{R}_{+}$, $\Tr_{A} [e^{t\mathcal{K}}(\omega_{A}\otimes \rho)] = e^{t\mathcal{L}}(\rho)$. Let $\mathcal{K} \in \algL$, taking the derivative of the previous equation and evaluating at $t=0$ we find
    \begin{equation}
        \mathcal{L}(\rho) = \Tr_{A}\left[\mathcal{K}(\omega_{A} \otimes \rho)\right] = \Tr_{A}\left\{-i[H, \omega_{A} \otimes \rho] + \Gamma(\mathcal{R}(\omega_{A}\otimes\rho) - \omega_{A}\otimes\rho) \right\}
    \end{equation}
    where the last equality follows from applying Lemma~\ref{lm:altchar} found in Appendix A. Simplifying the above equation we find
    \begin{equation}
        \mathcal{L}(\rho) = -i\left(\Tr_{A}(H\omega_{A}\otimes\rho) - \Tr_{A}(\omega_{A}\otimes\rho H)\right) - \Gamma(\mathcal{R}_{A}(\rho) - \rho)
    \end{equation}
    where the map $\mathcal{R}_{A} \in \mathbf{CPTP}$ and is defined by $\mathcal{R}(\rho) = \Tr_{A}\left[\mathcal{R}(\omega_{A}\otimes \rho)\right]$. Using Lemma~2 from Appendix B of \cite{vomende2023quantum} we find
    \begin{equation}
        \mathcal{L}(\rho) = -i[H_{A}, \rho] + \Gamma(\mathcal{R}_{A}(\rho) - \rho)
    \end{equation}
    where $H_{A}$ is defined as the unique operator satisfying $\Tr(H_{A}\rho) = \Tr(H(\omega_{A}\otimes\rho))$, guaranteed by Lemma~2 from Appendix B of \cite{vomende2023quantum}. Thus we find 
    \begin{equation}
        \mathcal{L} = -i\mathrm{ad}_{H_{A}} + \Gamma(\mathcal{R}_{A} + \mathcal{I}),
    \end{equation}
    which implies $\mathcal{L} \in \algL$ by Lemma~\ref{lm:altchar} from Appendix A.
\end{proof}

Finally, we may wonder if relaxing the embedding procedure described above can make $\algL$ universal. At the bare minimum we only require that at the final time of simulation the quantum channel induced on the subsystem matches the quantum channel generated by the targeted Lindbladian dynamics. More formally we require, for all Lindbladians $\mathcal{L}$ there exists a $\mathcal{K} \in \algL$ such that for some $\omega_{A}$
\begin{equation}
    \Tr_{A}[e^{t\mathcal{K}}(\omega_{A} \otimes \rho)] = e^{t\mathcal{L}}(\rho).
\end{equation}
Interestingly, the embedding procedure described above always exists due to the Stinespring dilation theorem, although this result is not as useful as suggested. The Stinespring dilation theorem always assures us there exists some unitary, generated by some Hamiltonian, which can implement an arbitrary quantum channel on a subsystem given the ancillary Hilbert space is large enough. Although this result tells us a $\mathcal{K} \in \algL$ exists which can satisfy the embedding procedure described above, it does not tell us \emph{how} to construct this $\mathcal{K}$ from the final time $t$ and target Lindbladian $\mathcal{L}$. In addition, this result only informs us that the embedding ability of Hamiltonian dynamics, a restricted subset of the dynamics generated by $\algL$, is universal. Given that an arbitrary Lindbladian $\mathcal{L} \in \algL$ may be expressed as a sum of a Hamiltonian portion and dissipator with jump operators satisfying equation~\eqref{eq:algconstraint}, one may wonder whether pure dissipators with jump operators satisfying equation~\eqref{eq:algconstraint} can also embed arbitrary Lindbladian dynamics. Such a conjecture can be proven false as an implication to the following Lemma.
\begin{lemma}
    \label{lm:globalprop}
    Let $\mathcal{D} \in \algL$ where the Hamiltonian $H=0$ when equation~\eqref{eq:algconstraint} is satisfied with $\Gamma \neq 0$. Then for all $t \in \mathbb{R}_{+} \setminus \{0\}$
    \begin{equation}
        \label{eq:specialform}
        e^{t\mathcal{D}} = e^{-t\Gamma}\mathcal{I} + (1 - e^{-t\Gamma})\Phi_{t}
    \end{equation}
    where $\Phi_{t} \in \mathbf{CPTP}$
\end{lemma}
\begin{proof}
    Let $\mathcal{D} \in \algL$ where the Hamiltonian $H=0$ when equation~\eqref{eq:algconstraint} is satisfied. Lemma~\ref{lm:altchar} in Appendix B implies we can write $\mathcal{D} = \Gamma (\mathcal{R} - \mathcal{I})$ where $\mathcal{R} \in \mathbf{CPTP}$. Thus we find
    \begin{equation}
        e^{t\mathcal{D}} = e^{t\Gamma(\mathcal{R} - \mathcal{I})} =  e^{-t\Gamma}\sum_{n=0}^{\infty}\frac{(t\Gamma)^{n}}{n!}\mathcal{R}^{\circ n} = e^{-t\Gamma}\left(\mathcal{I} + \sum_{n=1}^{\infty}\frac{(t\Gamma)^{n}}{n!}\mathcal{R}^{\circ n}\right) = e^{-t\Gamma}\mathcal{I} + (1 - e^{-t\Gamma})\Phi_{t}
    \end{equation}
    where we have defined the map $\Phi_{t} = (e^{t\Gamma} - 1)^{-1}\sum_{n=1}^{\infty}\frac{(t\Gamma)^{n}}{n!}\mathcal{R}^{\circ n }$. Since $\mathcal{R} \in \mathbf{CPTP}$ and $\Gamma \in \mathbb{R}_{+} \setminus \{0\}$, then for all $t \in \mathbb{R}_{+} \setminus \{0\}$, $\Phi_{t} \in \mathbf{CP}$. To prove trace preservation note,
    \begin{equation}
        \Phi_{t}^{\dagger}(\mathbb{I}) = (e^{t\Gamma} - 1)^{-1} \sum_{n=1}^{\infty}\frac{t\Gamma}{n!}\mathcal{R}^{\dagger \circ n}(\mathbb{I}) = (e^{t\Gamma} - 1)^{-1}\sum_{n=1}^{\infty}\frac{(t\Gamma)^{n}}{n!} \mathbb{I} = \mathbb{I}
    \end{equation}
    where the second equality follows from $\mathcal{R} \in \mathbf{CPTP}$. Thus $\Phi_{t} \in \mathbf{UCP}$ implying $\Phi_{t} \in \mathbf{CPTP}$.
\end{proof}
\begin{theorem}
    \label{cor:weakdiss}
    There exists a Lindbladian $\mathcal{L} \notin \algL$ such that there does not exist a $\mathcal{D} \in \algL$ where the Hamiltonian $H=0$ when equation~\eqref{eq:algconstraint} is satisfied and 
    \begin{equation}
        \Tr_{A}[e^{t\mathcal{D}}(\omega_{A}\otimes \rho)] = e^{t\mathcal{L}}(\rho).
    \end{equation}
\end{theorem}
\begin{proof}
    We proceed by proof by contradiction. Let $\mathcal{L}$ be the Lindbladian which generates the amplitude damping channel, represented by the jump operators $L = \ket{0}\bra{1}$ and Hamiltonian $H=0$. For all $t\in \mathbb{R}_{+}$, the channel $e^{t\mathcal{L}}$ has a Kraus representation given by 
    \begin{equation}
        \begin{cases}
            K_{0} = \ket{0}\bra{0} + \sqrt{1 - p(t)} \ket{1}\bra{1} \\
            K_{1} = \sqrt{p(t)}\ket{0}\bra{1}
        \end{cases}
    \end{equation}
    where $p(t) = 1 - e^{-t \Gamma}$. Note, the set of operators given by $\{K_{\mu}^{\dagger}K_{\nu}\}_{\mu,\nu = 0}^{1}$ is linearly independent for all $t \in \mathbb{R}_{+} \setminus \{0\}$ implying the channel $e^{t\mathcal{L}}$ is an extreme point in the set of all quantum channels acting on a single qubit. Assuming there exists a $\mathcal{D} \in \algL$ where the Hamiltonian $H=0$ when equation~\eqref{eq:algconstraint} is satisfied and $\Tr_{A}[e^{t\mathcal{D}}(\omega_{A}\otimes \rho)] = e^{t\mathcal{L}}(\rho)$, implies we can apply Lemma~\ref{lm:globalprop} and write
    \begin{align}
        e^{t\mathcal{L}}(\rho) &= \Tr_{A}\left[e^{t\mathcal{D}}(\omega_{A}\otimes\rho)\right] = \Tr_{A}\left[(e^{-t\Gamma}\mathcal{I} + (1 - e^{-t\Gamma})\Phi_{t})[\omega_{A}\otimes\rho]\right] \nonumber \\
                               &= e^{-t\Gamma}\rho + (1 - e^{-t\Gamma})\Tr_{A}\left[\Phi_{t}(\omega_{A}\otimes\rho)\right]
    \end{align}
    for all $t \in \mathbb{R} \setminus \{0\}$.
    Defining the map $\Phi_{t}^{(A)}(\rho) \in \mathbf{CPTP}$ as $\Phi_{t}^{(A)}(\rho)  = \Tr_{A}\left[\Phi_{t}(\omega_{A}\otimes\rho)\right]$ we can write 
    \begin{equation}
        e^{t\mathcal{L}} = e^{-t\Gamma} \mathcal{I} + (1 - e^{-t\Gamma}) \Phi_{t}^{(A)},
    \end{equation}
    which is a convex combination of single qubit quantum channels. Thus we have arrived at a contradiction given that $e^{t\mathcal{L}}$ has already been shown to be an extreme point in the set of quantum channels for $t\in\mathbb{R} \setminus \{0\}$.
\end{proof}

Given the results discussed above, the set of Lindbladians for which our algorithm can be applied to is quite rigid. Lemma~\ref{lm:subset} confirms our intuition, that our algorithms restriction does limit its applicability to a subset of all Lindbladians. In addition, Lemma~\ref{lm:noembed} and Theorem~\ref{cor:weakdiss} imply the ability of this subset to embed arbitrary Lindbladian dynamics is restricted as well. Such results suggest, that extending the applicability of the algorithm will likely involve modifying its structure as opposed to using it as a black-box subroutine.

\section{Discussion}
\label{sec:disc}

In this work we have presented a novel quantum algorithm for simulating time-independent Lindblad master equations inspired by the classical Monte-Carlo wavefunction algorithms for simulating Lindblad master equations. Interestingly, for the master equations it applies to, our algorithm achieves the additive scaling between $T$ and $\epsilon$ with respect to queries to the Lindbladian's jump operators. As far as we are aware, our algorithm is the first to achieve such a result. Heuristically, the key factor behind our algorithms improved jump operator query complexity comes from unraveling the Lindbladian dynamics as a stochastic process over Hilbert space. More specifically, because the unraveling chosen has quantum jump events driven by a Poisson process, the expected number of jump events tends to grow linearly with respect to the total simulation time, with the right-tail of the distribution decaying exponentially. Given that each jump event entails one query to the Lindbladian's jump operators, implies that the total number of queries tends to grow linearly with respect to the time interval, with a small additive perturbation occurring from the exponentially decaying right-tail of the distribution. This allows us to significantly minimize the query complexity of our algorithm to be as close to linear as possible.

Although our query complexity to the Lindbladian's jump operators is additive with respect to $T$ and $\epsilon$, the algorithm we have presented above is still sub-optimal across several areas. For one, it requires the Lindbladian to satisfy the constraint of equation~\eqref{eq:algconstraint}. This suggests lifting the constraint of equation~\eqref{eq:algconstraint} as a first avenue for improvement. Another future direction involves improving the algorithms query complexity to the Lindbladian's Hamiltonian. Interestingly, there is a mismatch between our algorithm's Hamiltonian query complexity and its jump operator query complexity. On one hand, our jump operator query complexity is additive between $T$ and $\epsilon$, but we sacrifice such scaling for the Hamiltonian query complexity. It would be interesting to explore whether the Hamiltonian query complexity can be improved or whether there exists a trade-off between the two resources.

\begin{acknowledgments}
This work is supported by Sandia National Laboratories’ Laboratory Directed Research and Development program (Contract \#2534192). Additional support by DOE's Express: 2023 Exploratory Research For Extreme-scale Science Program under Award Number DE-SC0024685 is acknowledged.

\vspace{3mm}
\textbf{Note Added:} This work was presented at QSim August 2025. Right before submitting this manuscript to Arxiv, a different quantum algorithm was announced to sample from Gibbs states with a Lindbladian simulation subroutine achieving similar scaling, but for only jump operators that are unitary and without considering Hamiltonian dynamics \cite{shang2025exponentiallindbladianfastforwarding}. The additive query complexity in \cite{shang2025exponentiallindbladianfastforwarding} follows from cleverly truncating the Taylor series for the exponential of the full Lindbladian while in our algorithm the additive complexity follows from simulating only the quantum trajectories associated with the Lindbladian and not the full Lindbladian.

\end{acknowledgments}

\bibliographystyle{apsrev4-2}
\bibliography{refs}

\clearpage

\setcounter{table}{0}
\renewcommand{\thetable}{S\arabic{table}}%
\setcounter{figure}{0}
\renewcommand{\thefigure}{S\arabic{figure}}%
\setcounter{section}{0}
\setcounter{equation}{0}
\renewcommand{\theequation}{S\arabic{equation}}%


\appendix

\section{Proof of Lemma 11.}
\label{app:A}
In this section we provide an alternative way to characterize the set of Lindbladians for which our algorithm can be applied to. As discussed in the main text our algorithm is applicable only to Lindbladians with jump operators $\{L_{\mu}\}_{\mu=1}^{m}$ which satisfy
\begin{equation}
    \label{eq:algconst2}
    \sum_{\mu=1}^{m}L_{\mu}^{\dagger}L_{\mu} = \Gamma \mathbb{I}
\end{equation}
for some $\Gamma \in \mathbb{R}_{+}$. Defining the set of Lindbladians which have representation satisfying the above equation as $\algL$ we proceed to prove the following Lemma giving us an alternative way to characterize Lindbladians which satisfy our algorithms constraint.
\begin{lemma}
    \label{lm:altchar}
    $\mathcal{L} \in \algL$ if and only if there exists an $\mathcal{R} \in \mathbf{CPTP}$ such that $\mathcal{L} = -i\mathrm{ad}_{H} + \Gamma(\mathcal{R} - \mathcal{I})$.
\end{lemma}
\begin{proof}
    Let $\mathcal{L} \in \algL$, choosing the representation of $\mathcal{L}$ where the jump operators satisfy equation~\eqref{eq:algconst2} we find
    \begin{align}
        \mathcal{L}(\rho) &= -i [H, \rho] + \sum_{\mu=1}^{m}\left((L_{\mu} \rho L_{\mu}^{\dagger} - \frac{1}{2}\{L_{\mu}^{\dagger}L_{\mu}, \rho \}\right) = -i[H , \rho ] + \sum_{\mu=1}^{m}L_{\mu}\rho L_{\mu}^{\dagger} - \frac{\Gamma}{2}\{\mathbb{I}, \rho\} \nonumber \\ 
                          &= -i [H , \rho ] + \Gamma\left(\mathcal{R} - \mathcal{I}\right),
    \end{align}
    where we have defined the map $\mathcal{R}(\rho) = \sum_{\mu=1}^{m}K_{\mu} \rho K_{\mu}^{\dagger} = \frac{1}{\Gamma} \sum_{\mu=1}^{m}L_{\mu} \rho L_{\mu}^{\dagger}$ and $K_{\mu} = L_{\mu} / \sqrt{\Gamma}$. Note,
    \begin{equation}
        \sum_{\mu=1}^{m}K_{\mu}^{\dagger}K_{\mu} = \frac{1}{\Gamma}\sum_{\mu=1}^{m}L_{\mu}^{\dagger}L_{\mu} = \mathbb{I}.
    \end{equation}
    Thus $\mathcal{R} \in \mathbf{CPTP}$, implying the forward direction.

    For the backwards direction let $\mathcal{L} = -i\mathrm{ad}_{H} + \Gamma\left(\mathcal{R} - \mathcal{I}\right)$ with $\mathcal{R} \in \mathbf{CPTP}$. Choosing a Kraus representation for $\mathcal{R}$ we find
    \begin{align}
        \mathcal{L}(\rho) &= -i [H , \rho ] + \Gamma\left(\sum_{\mu=1}^{m}K_{\mu}\rho K_{\mu}^{\dagger} - \rho  \right) = -i[H, \rho ] + \sum_{\mu=1}^{m}\Gamma K_{\mu}\rho K_{\mu}^{\dagger} - \frac{\Gamma}{2}\{\mathbb{I}, \rho \} \nonumber \\
                          &= -i [H, \rho ] + \sum_{\mu=1}^{m} L_{\mu} \rho L_{\mu}^{\dagger} - \frac{1}{2}\{\Gamma \mathbb{I}, \rho \},
    \end{align}
    where we have defined $L_{\mu} = \sqrt{\Gamma} K_{\mu}$. The above equation is manifestly a Lindbladian dissipator with jump operators satisfying 
    \begin{equation}
        \sum_{\mu=1}^{m}L_{\mu}^{\dagger}L_{\mu} = \Gamma \sum_{\mu=1}^{m} K_{\mu}^{\dagger}K_{\mu} = \Gamma \mathbb{I},
    \end{equation}
    thus $\mathcal{L} \in \algL$.
\end{proof}
Lemma~\ref{lm:altchar} allows a simple physical characterize Lindbladians in $\algL$ as those which are a sum of a Hamiltonian portion and a dissipator which can be expressed as proportional to a difference between applying a quantum channel and not.
\section{Proof of Theorem 13.}
\label{app:B}

In this section we provide a proof of the well known result that events which occur with holding times $t_{h}^{(i)}$ where $\mathrm{Pr}(t_{h}^{(i)} \leq \tau) = 1 - e^{-\Gamma \tau}$ have a counting process which is Poissonian. To begin we define the random variables $N(T)$ and $S^{(n)}$ according to our holding time random variables. Define the total waiting time for $n$ events to occur to be  
\begin{equation}
    S^{(n)} = \begin{cases}\sum_{i=1}^{n}t_{h}^{(i)} \text{ for } n > 0\\ 0  \text{ for } n = 0 \end{cases}.
\end{equation}
From $S^{(n)}$ we can define the total number of events occurring with-in a time interval $T$ to be 
\begin{equation}
    \label{eq:countingprocess}
    N(T) = \sup\{n : S^{(n)} \leq T\} = \sum_{n=0}^{\infty}\bm{1}\{S_{n} \leq T\}
\end{equation}
where $\bm{1}$ denotes the indicator function, returning one if event $S_{n} \leq T$ occurs and zero otherwise.

\begin{lemma}
    \label{lm:ntos}
    For any $T > 0$ and $n\in\mathbb{N}$, 
    \begin{equation}
        \mathrm{Pr}(N(T) = n) = \mathrm{Pr}(S^{(n)} \leq T) - \mathrm{Pr}(S^{(n+1)} \leq T)
    \end{equation}
\end{lemma}
\begin{proof}
    \begin{align}
        \mathrm{Pr}(N(T) = n) &= \mathrm{Pr}(\{N(T) \geq n\} \cap \{N(T) < n +1\}) \nonumber \\
                              &= \mathrm{Pr}(N(T) \geq n ) + \mathrm{Pr}(N(T) < n + 1) - \mathrm{Pr}(\{N(T) \geq n\} \cup \{N(T) < n +1\}) \nonumber \\
                              &= \mathrm{Pr}(N(T) \geq n ) - \mathrm{Pr}(N(T) \geq  n + 1) \nonumber \\
                              &= \mathrm{Pr}(S^{(n)} \leq T) - \mathrm{Pr}(S^{(n+1)} \leq T)
    \end{align}
    where the last line follows from the fact that the event $\{N(T) \geq n\}$ occurs only when the $n$th event occurs before some time $T$. Thus events $\{N(T) \geq n \}$ and  $\{S^{(n)} \leq T\}$ occur together.
\end{proof}

\begin{theorem}
    \label{thm:pois}
    Let $N(T)$ be defined according to equation~\eqref{eq:countingprocess} and $\mathrm{Pr}(t_{h}^{(i)} \leq \tau) = 1 - e^{-\Gamma \tau}$ for $\Gamma > 0$. Then $N(T) \sim \mathrm{Pois}(\Gamma T) = e^{-\Gamma T} \frac{(\Gamma T)^{N}}{N!}$.
\end{theorem}
\begin{proof}
    Consider the probability of waiting longer than $T$ for $n$ events to occur, thus we can write
    \begin{align}
        \mathrm{Pr}(S_{n} > T) &= \mathrm{Pr}(t_{h}^{(1)} > T) + \int_{0}^{T} \mathrm{Pr}\left(\sum_{i=2}^{n}t_{h}^{(i)} > T - t\right)\mathrm{Pr}(t_{h} = t)dt \nonumber \\
                               &= e^{-\Gamma T} + \int_{0}^{T} \Gamma e^{-\Gamma t} \mathrm{Pr}\left(\sum_{i=2}^{n}t_{h}^{(i)} > T - t\right) dt.
    \end{align}
    Defining $G_{n}(T) = \mathrm{Pr}(S_{n} > T)$, we can write the above equation as
    \begin{align}
        \label{eq:recur}
        G_{n}(T) &= e^{-\Gamma T} + \int_{0}^{T} \Gamma e^{-\Gamma t} G_{n-1}(T-t) dt.
    \end{align}
    Proposing the ansatz $G_{n}(T) = \sum_{a = 0}^{n-1}e^{-\Gamma T}\frac{(\Gamma T)^{a}}{a!}$, we can show it satisfies the recurrence relation above by plugging it into the right hand side of the above equation. We find
    \begin{align}
        G_{n}(T) &=  e^{-\Gamma T} + \int_{0}^{T}\Gamma e^{\Gamma t}\left(\sum_{a = 0}^{n-2}e^{-\Gamma (T - t)}\frac{(\Gamma (T-t))^{a}}{a!}\right) dt = e^{-\Gamma T} + \Gamma e^{-\Gamma T}\sum_{a = 0}^{n-2}\frac{\Gamma^{a}}{a!}\int_{0}^{T}(T - t)^{a}dt \nonumber \\
                 &= e^{-\Gamma T} + e^{-\Gamma T}\sum_{a = 0}^{n-2}\frac{\Gamma^{a+1}}{(a+1)!}T^{a+1} = \sum_{a = 0}^{n-1}e^{-\Gamma T}\frac{(\Gamma T)^{a}}{a!},
    \end{align}
    thus $G_{n}(T) = \sum_{a = 0}^{n-1}e^{-\Gamma T}\frac{(\Gamma T)^{a}}{a!}$ is solution to equation~\eqref{eq:recur}.

    Applying Lemma~\ref{lm:ntos} we find
    \begin{align}
        \mathrm{Pr}(N(T)=n) &= (1 - \mathrm{Pr}(S_{n} > T)) - (1 - \mathrm{Pr}(S_{n+1} > T)) = G_{n+1}(T) - G_{n}(T) \nonumber \\
                            &= e^{-\Gamma T} \frac{(\Gamma T)^{n}}{n!} = \mathrm{Pois}(\Gamma T)
    \end{align}
    thus concluding the proof.
\end{proof}

\end{document}